\documentclass[letter,11pt]{article}
\usepackage{amsmath} 
\usepackage{amssymb}  
\usepackage{amsthm}
\usepackage[dvipdfmx]{graphicx}
\usepackage{times}
\allowdisplaybreaks[3]
\newcommand{\Order}{\mathrm{O}}
\newcommand{\Prob}{\mathrm{Pr}}

\newcommand{\order}{\mathrm{o}}
\newcommand{\poly}{\mathrm{poly}}
\newcommand{\defeq}{\stackrel{\mbox{\scriptsize{\normalfont\rmfamily def}}}{=}}

\newcommand{\dd}{\mathrm{d}}

\renewcommand{\Vec}[1]{\mbox{\boldmath $#1$}}
\newcommand{\Real}{\mathbb{R}}
\newcommand{\Integer}{\mathbb{Z}}
\newcommand{\Vol}{\mathrm{Vol}}
\newcommand{\conv}{\mathrm{conv}}

\theoremstyle{plain}
\newtheorem{theorem}{Theorem}[section]
\newtheorem{lemma}[theorem]{Lemma}
\newtheorem{corollary}[theorem]{Corollary}
\newtheorem{proposition}[theorem]{Proposition}

\newtheorem{observation}[theorem]{Observation}
\newtheorem{claim}{Claim}

\theoremstyle{definition}
\newtheorem{algorithm}{Algorithm}
\newtheorem{problem}{Problem}

\title{
 An FPTAS for the Volume of a ${\cal V}$-polytope\\
 ---It is Hard to Compute The Volume of\\ The Intersection of Two Cross-polytopes}
\author{
 Ei Ando\thanks{Sojo University, 4-22-1, Ikeda, Nishi-Ku, Kumamoto, 860-0082, Japan. {ando-ei@cis.sojo-u.ac.jp}} \and 
 Shuji Kijima\thanks{
 Kyushu University, 744, Motooka, Nishi-Ku, Fukuoka, 819-0395, Japan.
 {kijima@inf.kyushu-u.ac.jp}}
}

\begin{document}
\maketitle
\begin{abstract}
Given an $n$-dimensional convex body by a membership oracle in general, 
  it is known that 
   any polynomial-time {\em deterministic} algorithm 
   cannot approximate its volume within ratio $(n/\log n)^n$. 
 There is a substantial progress on {\em randomized} approximation 
   such as Markov chain Monte Carlo 
   for a high-dimensional volume, and for many {\#}P-hard problems, 
 while some deterministic approximation algorithms are recently developed only for a few {\#}P-hard problems. 
 Motivated by a {\em deterministic} approximation of the volume of a ${\cal V}$-polytope, 
   that is a polytope with 
   few vertices and (possibly) exponentially many facets, 
  this paper investigates the volume of a ``knapsack dual polytope,'' 
  which is known to be {\#}P-hard due to Khachiyan (1989). 
 We reduce an approximate volume of a knapsack dual polytope 
  to that of the {\em intersection of two cross-polytopes}, 
  and give FPTASs for those volume computations. 
 Interestingly, 
  the volume of the intersection of two cross-polytopes (i.e., $L_1$-balls) is {\#}P-hard, 
  unlike the cases of $L_{\infty}$-balls or $L_2$-balls.  

\end{abstract}

\section{Introduction}
\subsection{Approximation of a high dimensional volume: randomized vs.\ deterministic}
 A high dimensional volume is hard to compute, even for approximation. 
 When an $n$-dimensional convex body is given by a {\em membership oracle}, 
  no polynomial-time {\em deterministic} algorithm can approximate 
  its volume within ratio $(n/\log n)^n$~\cite{BF1987,Elekes1986,Lovasz1986,PNAS2013}D
 Intuitively, 
  the impossibility comes from the fact 
   that the volume of an $n$-dimensional $L_{\infty}$-ball (i.e., hypercube) 
   is exponentially large to the volume of its inscribed $L_2$-ball or $L_1$-ball, 
  nevertheless the $L_2$-ball ($L_1$-ball as well) is convex and 
   touches each facet of the $L_{\infty}$-ball (see e.g., \cite{Matousek}). 
 Lov\'{a}sz said in~\cite{Lovasz1986} for a convex body $K$ that 
 ``{\it If $K$ is a {\em polytope}, 
   then there may be much better ways to compute $\Vol(K)$}.''
 Unfortunately,  
  an exact volume is often {\#}P-hard, 
  even for a relatively simple polytope. 
 For instance, 
  the {\em volume} of a knapsack polytope, 
   which is given by a box constraint (i.e., hypercube $[0,1]^n$) and a single linear inequality, 
  is a well-known {\#}P-hard problem~\cite{DF1988}. 

 The difficulty 
   caused by the exponential gap between $L_{\infty}$-ball and $L_{1}$-ball 
   also does harm a simple Monte Carlo algorithm. 
 Then, the Markov chain Monte Carlo (MCMC) method, a sophisticated {\em randomized} algorithm, 
   achieves a great success for approximating a high volume.  
 Dyer, Frieze and Kannan~\cite{DFK1991} gave the first 
  fully polynomial-time randomized approximation scheme (FPRAS) 
  for the volume computation of a general convex body\footnote{
   Precisely, they are concerned with a ``well-rounded'' convex body, 
   after an affine transformation of a general finite convex body. 
}. 
 They employed a {\em grid-walk}, 
   which is efficiently implemented with a membership oracle, and 
  showed its rapidly mixing, then 
  they gave an FPRAS runs in $\Order^*(n^{23})$ time 
   where $\Order^*$ ignores $\poly(\log n)$ and $1/\epsilon$ terms. 
 After several improvements, 
  Lov\'{a}sz and Vempala~\cite{LV2006} improved the time complexity to $\Order^*(n^4)$
  in which they employ hit-and-run walk, and recently 
  Cousins and Vempala~\cite{CV2015} gave an $\Order^*(n^3)$-time algorithm.  
 Many randomized techniques, including MCMC, also have been developed for designing FPRAS for {\#}P-hard problems. 

 In contrast, 
  a development of a {\em deterministic} approximation for {\#}P-hard problems is a current challenge, and 
   not many results seem to be known.  
 A remarkable progress is the {\em correlation decay} argument due to Weitz~\cite{Weitz2006}; 
  he designed a {\em fully polynomial time approximation scheme} ({\em FPTAS}\/)   
  for counting independent sets in graphs whose maximum degree is at least $5$. 
 A similar technique is independently presented by Bandyopadhyay and Gamarnik~\cite{BG08}, and 
  there are several recent developments on the technique, 
   e.g.,~\cite{GK2007,BGKNT2007,LLY12,LLY13,LLL14}. 
 For counting knapsack solutions\footnote{
  Given $\Vec{a} \in \mathbb{Z}_{>0}^n$ and $b \in \mathbb{Z}_{>0}$, 
   the problem is to compute $|\{ \Vec{x} \in \{0,1\}^n \mid \sum_{i=1}^n a_i x_i \leq b \}|$. 
  Remark that it is computed in polynomial time when all the inputs $a_i$ ($i=1,\ldots,n$) and $b$ 
  are bounded by $\poly(n)$, using a version of the standard dynamic programming for knapsack problem 
  (see e.g.,~\cite{Dyer,GKMSVV2011}). 
 Nevertheless, it should be worth noting that 
  \cite{GKM10} and \cite{SVV2012} 
needed special techniques, 
   different from ones for optimization problems, 
  to design FPTASs for the counting problem. 
}, 
  Gopalan, Klivans and Meka~\cite{GKM10}, and 
 \v{S}tefankovi\v{c}, Vempala and Vigoda~\cite{SVV2012} gave 
   deterministic approximation algorithms based on 
the dynamic programming (see also~\cite{GKMSVV2011}), 
   in a similar way to a simple random sampling algorithm by Dyer~\cite{Dyer}. 
 Modifying the dynamic programming, 
  Li and Shi~\cite{LS2014} gave an FPTAS 
   for the volume of a knapsack polytope, 
  which runs in $\Order((n^3/\epsilon^2) \poly \log b)$ time 
  where $b$ is the capacity of a knapsack.  
 Motivated by a different approach,  
  Ando and Kijima~\cite{AK2015} gave another FPTAS 
  for the volume of a knapsack polytope. 
 Their scheme is based on a classical approximate convolution, 
  and runs in $\Order(n^3/\epsilon)$ time, 
  independent of the size of items and the capacity of a knapsack 
  reckoning without numerical calculus.

\subsection{${\cal H}$-polytope and ${\cal V}$-polytope}
 An ${\cal H}$-polyhedron is an intersection of finitely many closed half-spaces in $\mathbb{R}^n$. 
 An ${\cal H}$-polytope is a bounded ${\cal H}$-polyhedron. 
 A ${\cal V}$-polytope is a convex hull of a finite point set in $\mathbb{R}^n$~\cite{Matousek}. 
 From the view point of computational complexity, 
  a major difference between an ${\cal H}$-polytope and a ${\cal V}$-polytope 
 is the measure of their `input size.' 
 An ${\cal H}$-polytope given by linear inequalities defining half-spaces  
   may have vertices exponentially many to the number of the inequalities, 
  e.g., 
    an $n$-dimensional hypercube is given by $2n$ linear inequalities 
    as an ${\cal H}$-polytope, and has $2^n$ vertices. 
 In contrast, a ${\cal V}$-polytope given by a point set 
   may have facets exponentially many to the number of vertices, 
  e.g., 
   an $n$-dimensional cross-polytope (that is an $L_1$-ball, in fact) is 
   given by a set of $2n$ points 
    as a ${\cal V}$-polytope, and it has $2^n$ facets. 

  There are many interesting properties, that are known, or unknown,
  between ${\cal H}$-polytope and ${\cal V}$-polytope~\cite{Matousek}. 
 A membership query is polynomial time for both ${\cal H}$-polytope and ${\cal V}$-polytope. 
 It is still unknown 
  about the complexity of a query if a given pair of 
   ${\cal V}$-polytope and ${\cal H}$-polytope are identical. 
 Linear programming (LP) on a ${\cal V}$-polytope is trivially polynomial time 
   since it is sufficient to check the objective value of all 
   vertices and hence 
   LP is usually concerned with an ${\cal H}$-polytope. 

\subsection{Volume of ${\cal V}$-polytope}
 Motivated by a hardness of the volume computation of a ${\cal V}$-polytope, 
  Khachiyan~\cite{Khachiyan1989} is concerned with the following ${\cal V}$-polytope:  
 Suppose a vector 
  $\Vec{a}=(a_1,\dots,a_n)\in \Integer_{\geq 0}^n$ is given,  
  where without loss of generality we may assume that $a_1 \geq a_2 \geq \cdots \geq a_n$. 
 Then let  
\begin{eqnarray}
 P_{\Vec{a}} \defeq \conv\left\{\pm \Vec{e}_1, \ldots, \pm \Vec{e}_n,\Vec{a}\right\}
\label{def:P_a}
\end{eqnarray}
  where $\Vec{e}_1,\ldots,\Vec{e}_n$ are the standard basis vectors in $\mathbb{R}^n$. 
 This paper calls $P_{\Vec{a}}$ {\em knapsack dual polytope}\footnote{
 See \cite{Matousek} for the duality of polytopes. 
 In fact, $P_{\Vec{a}}$ itself is not the dual of a knapsack polytope in a canonical form, 
  but it is obtained by an affine transformation from a dual of knapsack polytope under some assumptions. 
 Khachiyan~\cite{Khachiyan1993} says that 
   computing {\em $\Vol(P_{\Vec{a}})$ 
  `is ``polar'' to determining the volume of the intersection of a cube and a halfspace}.' 
}. 
 Khachiyan~\cite{Khachiyan1989} 
  showed that computing $\Vol(P_{\Vec{a}})$ is $\#P$-hard\footnote{
   If all $a_i$ ($i=1,\ldots,n$) are bounded by $\poly(n)$, it is computed in polynomial time, 
   so did the counting knapsack solutions. 
  See also footnote 1 for counting knapsack solutions. 
}.
 The hardness is given by a Cook reduction from counting set partitions, 
  of which the decision version is a cerebrated {\em weakly} NP-hard problem. 
 We do not know any (efficient) technique to translate the volume 
   between them a polytope and its dual polytope.

\subsection{Contribution}
 Motivated by a development of techniques 
  for {\em deterministic} approximation of the volumes of ${\cal V}$-polytopes, 
   this paper investigates the knapsack dual polytope $P_{\Vec{a}}$ given by \eqref{def:P_a}. 
 The main goal of the paper is to establish the following theorem. 
\begin{theorem}\label{th:main}
 For any $\epsilon$ ($0<\epsilon<1$), 
  there exists a {\em deterministic} algorithm that outputs a value $\widehat{V}$ 
  satisfying $(1-\epsilon)\Vol(P_{\Vec{a}}) \le \widehat{V} \le (1+\epsilon)\Vol(P_{\Vec{a}})$ 
  in $\Order(n^{10} \epsilon^{-6})$ time.
\end{theorem}
 As far as we know, 
   this is the first result on designing an FPTAS 
  for the volume of a ${\cal V}$-polytope which is known to be {\#}P-hard. 
 We also discuss some topics related to the volume of ${\cal V}$-polytopes 
  appearing in the proof process. 
 Let us briefly explain the outline of the paper.

\paragraph{Technique/organization}
 The first step for Theorem~\ref{th:main} 
  is a transformation of the {\em approximation problem} to another one: 
 An approximate volume of $P_{\Vec{a}}$ is reduced 
   to the volume of a union of geometric sequence of cross-polytopes (Section~\ref{sec:geometric-series}),  
  and then it is reduced to the volume of the intersection of two cross-polytopes (Section~\ref{sec:red-two-balls}). 
 We remark 
  that the former reduction is just for approximation, and is useless for a {\#}P-hardness. 
 A technical point of this step is that 
  the latter reduction is based on a subtraction---if 
  you are familiar with an approximation,  
  you may worry that a subtraction may destroy an approximation ratio\footnote{
 Suppose you know that $x$ is approximately 49 within 1{\%} error. 
 Then, you know that $x+50$ is approximately 99 within 1{\%} error. 
 However, it is difficult to say $50-x$ is approximately 1. 
 Even when additionally you know that $x$ does not exceed 50, 
  $50-x$ may be 2, 1, 0.1 or smaller than 0.001, meaning that the approximation ratio is unbounded. 
}. 
 It requires careful tuning of a parameter ($\beta$ in Section~\ref{sec:fptas-P_a}) 
    which plays conflicting functions in Sections~\ref{sec:geometric-series} and~\ref{sec:red-two-balls}: 
  the larger $\beta$, the better approximation in Section~\ref{sec:geometric-series}, 
  while the smaller $\beta$, the better in Section~\ref{sec:red-two-balls}. 
 Then, Section~\ref{sec:algorithm-P_a} claims by giving an appropriate $\beta$ that 
  if we have an FPTAS for the volume of an 
  {\em intersection of two cross-polytopes} 
  then we have an FPTAS of $\Vol(P_{\Vec{a}})$. 

 Section~\ref{sec:fptas-two-cp} is a technical core of the paper, 
  where we give an FPTAS for the volume of the intersection of two cross-polytopes (i.e., $L_1$-balls). 
 The scheme is based on a modified version of the technique developed in~\cite{AK2015}, 
   which is based on a classical approximate convolution. 
 At a glance, 
  the volume of the intersection of two-balls may seem easy. 
 It is true for two $L_{\infty}$-balls (i.e., hypercubes\footnote{
  To be precise, an $L_{\infty}$-ball is a hypercube in a position parallel to the axis, 
   meaning that any $L_{\infty}$-ball is transformed to any other one 
   by scaling and parallel move, without using a rotation. 
  If two hypercubes are not in a parallel position, 
   the volume of the intersection is {\#}P-hard since the volume of a knapsack polytope is. 
 }), or $L_2$-balls (i.e., Euclidean balls). 
 However, 
  we show in Section~\ref{sec:hardness} that  
  the volume of the intersection of cross-polytopes is {\#}P-hard. 
 Intuitively, 
  this interesting fact may come 
  from the fact that 
  the ${\cal V}$-polytope, meaning that an $n$-dimensional cross-polytope, 
   has $2^n$ facets. 
 In Section~\ref{sec:extension}, 
  we extend the technique in Section~\ref{sec:fptas-two-cp} 
  to the intersection of any constant number of cross-polytopes. 
 Section~\ref{sec:n+c} briefly discusses 
  the complexity of the volume computation of a ${\cal V}$-polytope 
  regarding the number of vertices.

\section{Preliminary}\label{sec:prelimiary}
 This section presents some notation. 
 Let $\conv(S)$ denote the convex hull of $S \subseteq \mathbb{R}^n$, 
  where $S$ is not restricted to a finite point set. 
 A {\em cross-polytope} $C(\Vec{c},r)$ 
  of radius $r \in \mathbb{R}_{> 0}$ centered at $\Vec{c} \in \mathbb{R}^n$ is given by 
\begin{eqnarray}
 C(\Vec{c},r) \defeq \conv\{ \Vec{c} \pm r\Vec{e}_i \ i=1,\ldots,n\}
\label{def:C(c,r)}
\end{eqnarray}
  where $\Vec{e}_1,\ldots,\Vec{e}_n$ are the standard basis vectors in $\mathbb{R}^n$. 
 Clearly, $C(\Vec{c},r)$ has $2n$ vertices. 
 In fact, $C(\Vec{c},r)$ is an $L_1$-ball in $\mathbb{R}^n$ described by 
\begin{eqnarray}
 C(\Vec{c},r) 
  &=&
   \left\{\Vec{x}\in \Real^n \,\left|\, \|\Vec{x}-\Vec{c}\|_1 \le r \right.\right\} 
 \label{def:C(c,r)2}\\
  &=&
  \left\{\Vec{x}\in \Real^n \mid \langle \Vec{x}-\Vec{c}, \Vec{\sigma}\rangle \le r 
  \ (\forall \Vec{\sigma}\in \{-1,1\}^n)\right\} 
 \label{def:C(c,r)3}
\end{eqnarray}
 where 
  $\|\Vec{u}\|_1 = \sum_{i=1}^n |u_i|$ for $\Vec{u}=(u_1,\ldots,u_n) \in \mathbb{R}^n$ and 
  $\langle \Vec{u},\Vec{v} \rangle = \sum_{i=1}^n u_iv_i$  
 for $\Vec{u},\Vec{v} \in \mathbb{R}^n$. 
 Note that $C(\Vec{c},r)$ has $2^n$ facets. 
 It is not difficult to see that 
  the volume of a cross-polytope in $n$-dimension is 
\begin{eqnarray}
 \Vol(C(\Vec{c},r)) = \frac{2^n}{n!} r^n \label{eq:vol-cp}
\end{eqnarray}
 for any $r \geq 0$ and $\Vec{c} \in \mathbb{R}^n$, 
  where $\Vol(S)$ for $S \subseteq \mathbb{R}^n$ denotes the ($n$-dimensional) volume of $S$.

\section{FPTAS for Knapsack Dual Polytope}
\label{sec:fptas-P_a}
 This section reduces an approximation of $\Vol(P_{\Vec{a}})$ to 
   that of the intersection of two cross-polytopes. 
 In Section~\ref{sec:fptas-two-cp}, 
   we will give an FPTAS for the volume of a latter polytope, accordingly 
   we obtain Theorem~\ref{th:main}. 

\subsection{Reduction to a geometric series of cross-polytopes}\label{sec:geometric-series}
 Let $\beta$ be a parameter\footnote{
     We will set $\beta = 1-\dfrac{\epsilon}{2n\|\Vec{a}\|_1}$, later.} 
  satisfying $0<\beta<1$, and 
  let $Q_0,Q_1,Q_2,\ldots$ be a sequence of cross-polytopes defined by 
\begin{eqnarray}
 Q_k \defeq C((1-\beta^k)\Vec{a}, \beta^k) \label{def:Q_k}
\end{eqnarray}
  for $k=0,1,2,\ldots$. 
Remark that 
\begin{eqnarray*}
 Q_0 &=& C(\Vec{0}, 1), \\
 Q_1 &=& C((1-\beta)\Vec{a}, \beta), \\
 Q_{\infty} &=& C(\Vec{a}, 0) = \{\Vec{a}\}. 
\end{eqnarray*}
 The goal of Section~\ref{sec:geometric-series} is to establish the following. 
\begin{lemma}\label{lem:geometric-series}
 Let $\epsilon$ satisfy $0<\epsilon<1$. 
 If $1-\beta \leq \dfrac{c_1\epsilon}{n\|\Vec{a}\|_1}$ where $0 < c_1 \epsilon < 1$, 
  then 
\begin{eqnarray*}
 (1-c_1\epsilon)\Vol(P_{\Vec{a}}) 
 \leq
 \Vol\left(\bigcup_{k=0}^\infty Q_k\right)
  \leq \Vol(P_{\Vec{a}}). 
\end{eqnarray*}
\end{lemma}
 We remark that $P_{\Vec{a}}$ defined by~\eqref{def:P_a} is also described by 
\begin{eqnarray}
  P_{\Vec{a}} = \conv(C(\Vec{0},1) \cup \{\Vec{a}\}) 
\label{def:P_a2}
\end{eqnarray}
 using $C(\Vec{0},1)$. 
  Figure~\ref{fig:crosspolytope_infseq} illustrates the approximation of 
  $P_{\Vec{a}}$ by this infinite sequence of cross-polytopes. 
\begin{figure}[t]
\begin{center}
  \includegraphics[clip, width=160pt]{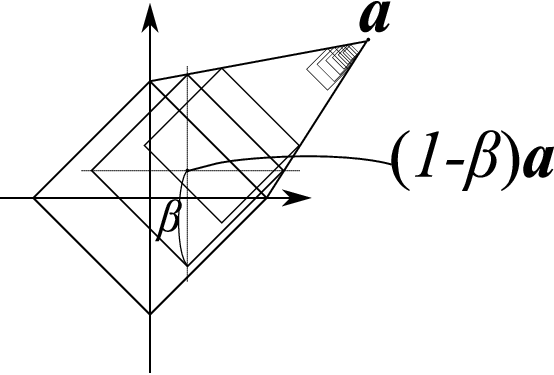}
\end{center}
\caption{Approximating $P_{\Vec{a}}$ by an infinite sequence of cross-polytopes. }
\label{fig:crosspolytope_infseq}
\end{figure}
 The second inequality in Lemma~\ref{lem:geometric-series} 
  is relatively easy by the following lemma. 
\begin{lemma}\label{lemma:subset}
\begin{align*}
 \bigcup_{k=0}^\infty Q_k\subseteq P_{\Vec{a}}.
\end{align*}
\end{lemma}
\begin{proof}
 Notice that $P_{\Vec{a}}$ is the convex hull of $Q_0$ and $\Vec{a}$ 
    (see~\eqref{def:P_a} for the definition of $P_{\Vec{a}}$). 
 We give a map $\eta_k \colon Q_k \to Q_0$, and show that 
  any $\Vec{x} \in Q_k$ is in the line segment between $\eta_k(\Vec{x})$ and $\Vec{a}$. 
 Notice that both $Q_k$ and $Q_0$ are $L_1$-balls, meaning that they are similar, 
  and our map $\eta_k$ is a natural correspondence between them. 
 Let 
\begin{eqnarray}
 \eta_k(\Vec{x}) 
  = \frac{\Vec{x}-(1-\beta^k)\Vec{a}}{\beta^k}.
\label{eq:eta_k}
\end{eqnarray}
 Then, $\|\eta_k(\Vec{x}) -\Vec{0}\|_1 \leq 1$ holds 
  since $\|\Vec{x}-(1-\beta^k)\Vec{a}\|_1 \leq \beta^k$ 
  by the assumption that $\Vec{x} \in Q_k$ (recall the definition \eqref{def:Q_k} of $Q_k$). 
 This implies that $\eta_k(\Vec{x}) \in Q_0$. 
 Notice that \eqref{eq:eta_k} implies that 
\begin{eqnarray*}
 \Vec{x} = \beta^k \eta_k(\Vec{x}) + (1-\beta^k)\Vec{a}
\end{eqnarray*}
 where $0 < \beta^k < 1$, 
 meaning that $\Vec{x}$ is given by a convex combination of $\eta_k(\Vec{x})$ and $\Vec{a}$. 
\end{proof}

 Next, we show the first inequality in Lemma~\ref{lem:geometric-series}. 
 As a preliminary, we show the following. 
\begin{lemma}\label{lemma:supset}
\begin{eqnarray*}
 \bigcup_{k=0}^{\infty} \conv (Q_k \cup Q_{k+1}) \cup \{\Vec{a}\} \supseteq P_{\Vec{a}}.
\end{eqnarray*}
\end{lemma}
 In fact, Lemma~\ref{lemma:supset} holds by equality, but we only show $\supseteq$ here. 
\begin{proof}
 Suppose $\Vec{x} \in P_{\Vec{a}}$. 
 Since $P_{\Vec{a}} = \conv(C(\Vec{0},1) \cup \{\Vec{a}\})$, 
  it is not difficult to see that 
  there is a $\Vec{y}_0 \in C(\Vec{0},1)$ such that 
  $\Vec{x}$ is in the line segment $\overline{\Vec{a},\Vec{y}}$ between $\Vec{a}$ and $\Vec{y}$. 
 Using bijective maps $\eta_k$ for $k=1,2,\ldots$ defined by \eqref{eq:eta_k}, 
  $\eta_k^{-1}(\Vec{y})$ is in $\overline{\Vec{a},\Vec{y}}$ in order of $k$. 
 Suppose $\Vec{x}$ is between $\eta_{k^*}^{-1}(\Vec{y})$ and $\eta_{k^*+1}^{-1}(\Vec{y})$, 
  then $\Vec{x} \in \conv(Q_{k^*} \cup Q_{k^*+1})$. 
 We obtain the claim. 
\end{proof}

\begin{lemma}\label{lemma:upper}
 If $1-\beta \leq \dfrac{c_1\epsilon}{n \|\Vec{a}\|_1}$, then  
\begin{align*}
 \Vol\left(\bigcup_{k=0}^\infty Q_k\right) \geq (1-c_1\epsilon)\Vol(P_{\Vec{a}}). 
\end{align*}
\end{lemma}
\begin{proof}
 For convenience, let 
\begin{eqnarray}
 \Delta &\defeq& (1-\beta)\|\Vec{a}\|_1 \label{eq:Delta}\\
 &\leq& \dfrac{c_1\epsilon}{n \|\Vec{a}\|_1}\|\Vec{a}\|_1 
  =  \dfrac{c_1\epsilon}{n}  \nonumber
\end{eqnarray}
 Notice that $\Delta$ is the distance between the centers of cross-polytopes $Q_0$ and $Q_1$. 
 Let 
\begin{eqnarray*}
 R_k  &\defeq& C((1-\beta^k)\Vec{a},(1+\Delta)\beta^k)
\end{eqnarray*}
 e.g., $R_0 = C(\Vec{0},1+\Delta)$. 
 We claim that $R_k \supseteq \conv (Q_k \cup Q_{k+1})$, 
  which implies $\bigcup_{k=0}^{\infty} R_k \cup \{\Vec{a}\} \supseteq P_{\Vec{a}}$ by Lemma~\ref{lemma:supset}. 
 It is not difficult to see that 
  $R_k  = C((1-\beta^k)\Vec{a},(1+\Delta)\beta^k) \supseteq C((1-\beta^k)\Vec{a},\beta^k) = Q_k$. 
 Next, we show that 
  $R_k  = C((1-\beta^k)\Vec{a},(1+\Delta)\beta^k) 
   \supseteq C((1-\beta^{k+1})\Vec{a},\beta^{k+1}) = Q_{k+1}$.  
 Let 
  $\Vec{x} \in Q_{k+1}$, then 
\begin{eqnarray*}
 \|\Vec{x} - (1-\beta^{k+1})\Vec{a}\|_1
 &=&
  \|\Vec{x} - (1-\beta^k)\Vec{a} - (\beta^k -\beta^{k+1})\Vec{a} \|_1 \\
 &\geq& 
  \|\Vec{x} - (1-\beta^k)\Vec{a}\|_1 - (\beta^k -\beta^{k+1})\|\Vec{a} \|_1
\end{eqnarray*}
   and it implies that 
\begin{eqnarray*}
 \|\Vec{x} - (1-\beta^k)\Vec{a}\|_1 
 &\leq & 
  \beta^{k+1} + (\beta^k -\beta^{k+1})\|\Vec{a} \|_1 \\
 &=& 
   \beta^{k+1} + \beta^k(1 -\beta)\|\Vec{a} \|_1 \\
 &=& 
   \beta^{k+1} + \beta^k \Delta \\
 &\leq &
  (1+\Delta)\beta^k. 
\end{eqnarray*}
 Thus $R_k \supseteq Q_{k+1}$. 
 Clearly a cross-polytope $R_k$ is convex, we obtain the claim $R_k \supseteq \conv(Q_k \cup Q_{k+1})$, 
  which implies that 
 $\Vol(\bigcup_{k=0}^{\infty} R_k) \geq \Vol(P_{\Vec{a}})$ as we prescribed. 

 Then, 
  we bound the ratio 
  $\Vol(\bigcup_{k=0}^{\infty} Q_k)/\Vol(P_{\Vec{a}})$ by 
  $\Vol(\bigcup_{k=0}^{\infty} Q_k)/\Vol(\bigcup_{k=0}^{\infty} R_k)$. 
 For convenience, let 
\begin{eqnarray*}
 \hat{R}_k 
  &\defeq& 
  C((1+\Delta)(1-\beta^k)\Vec{a},(1+\Delta)\beta^k). 
\end{eqnarray*}
 Clearly, $\Vol(\hat{R}_k) = \Vol(R_k)$. 
 It is not difficult to observe that $\Vol (\hat{R}_k \cap \hat{R}_{k+1}) \leq \Vol (R_k \cap R_{k+1})$, 
  which implies that 
  $\Vol(\bigcup_{k=0}^{\infty} \hat{R}_k) \geq \Vol(\bigcup_{k=0}^{\infty} R_k)$. 
 Furthermore, 
\begin{eqnarray}
  \Vol\left(\bigcup_{k=0}^{\infty} \hat{R}_k\right) 
  = (1+\Delta)\Vol\left(\bigcup_{k=0}^{\infty} Q_k\right)
\label{tmp20160403a}
\end{eqnarray}
 holds since $\bigcup_{k=0}^{\infty} \hat{R}_k$ and $\bigcup_{k=0}^{\infty} Q_k$ are similar. 
 Consequently, 
\begin{align*}
 \frac{\Vol\left(\bigcup_{k=0}^{\infty} Q_k\right)}{\Vol(P_{\Vec{a}})}
 &\geq 
 \frac{\Vol\left(\bigcup_{k=0}^{\infty} Q_k\right)}{\Vol\left(\bigcup_{k=0}^{\infty} R_k\right)} 
  \geq 
 \frac{\Vol\left(\bigcup_{k=0}^{\infty} Q_k\right)}{\Vol\left(\bigcup_{k=0}^{\infty} \hat{R}_k\right)} \\
 &=  \frac{1}{(1+\Delta)^n} 
  && (\mbox{by \eqref{tmp20160403a}}) \\
 &=  \frac{1}{(1+(1-\beta)\|\Vec{a}\|_1)^n} 
  && (\mbox{by \eqref{eq:Delta}}) \\
 &\geq  \frac{1}{\left(1+\dfrac{c_1\epsilon}{n\|\Vec{a}\|_1}\|\Vec{a}\|_1\right)^n} 
  && \left(\mbox{since } 1-\beta \leq \frac{c_1\epsilon}{n\|\Vec{a}\|_1} \mbox{ (hypo.)}\right)
\\
 &=  \frac{1}{\left(1+\dfrac{\epsilon}{n}\right)^n} \\
 &\geq  \left(1-\dfrac{\epsilon}{n}\right)^n \\
 &\geq  1-\epsilon. 
\end{align*} 
We obtain the claim. 
\end{proof}

 Lemma~\ref{lem:geometric-series} follows Lemmas~\ref{lemma:subset} and \ref{lemma:upper}.

\subsection{Reduction to the intersection of two cross-polytopes}\label{sec:red-two-balls}
\subsubsection{The volume of $\bigcup_{k=0}^{\infty} Q_k$}\label{sec:red-two-ballsA}
 Section~\ref{sec:red-two-ballsA} claims the following. 
\begin{lemma}\label{lem:vol-union-gem-series}
\begin{eqnarray*}
 \Vol\left(\bigcup_{k=0}^{\infty} Q_k\right)
 = 
 \frac{1}{1-\beta^n} \left(\frac{2^n}{n!} -\Vol(Q_1 \cap Q_0)\right). 
\end{eqnarray*}
\end{lemma}
 The first step of the proof is the following recursive formula. 
\begin{lemma}\label{lem:vol-union-gem-seriesA}
\begin{eqnarray*}
 \bigcup_{k=0}^{m} Q_k 
  &=& \left(\dot\bigcup_{k=0}^{m-1} Q_k \setminus Q_{k+1}\right) \,\dot\cup\, Q_m
\end{eqnarray*}
 where $A \,\dot\cup\, B$ denotes the disjoint union of $A$ and $B$, 
  meaning that $A \,\dot\cup\, B = A \cup B$ and $A \cap B = \emptyset$. 
\end{lemma}
 Lemma~\ref{lem:vol-union-gem-seriesA} is seemingly trivial, where 
 the point is the following claim. 
 
\begin{claim}\label{claim}
\begin{eqnarray*}
 \left(\bigcup_{k=0}^{m} Q_k\right) \cap Q_{m+1} 
 =  Q_m \cap Q_{m+1}.
\end{eqnarray*}
\end{claim}
\begin{proof}[Proof of Claim~\ref{claim}]
 The inclusion ``$\supseteq$'' is clear. 
 We prove the other inclusion ``$\subseteq$.'' 
Suppose 
  for an arbitrary $k \in \{0,1,\ldots,m-1\}$  
  that $\Vec{x} \in Q_k \cap Q_{m+1}$ holds. 
Then $\Vec{x} \in Q_k$ implies 
 $\langle \Vec{x} - (1-\beta^k) \Vec{a}, \Vec{\sigma} \rangle 
 =\langle \Vec{x} - \Vec{a}, \Vec{\sigma} \rangle + \langle \beta^k \Vec{a}, \Vec{\sigma} \rangle \leq \beta^k$
 holds for any $\Vec{\sigma} \in \{-1,1\}^n$, and 
 $\Vec{x} \in Q_{m+1}$ implies 
 $\langle \Vec{x} - (1-\beta^{m+1}) \Vec{a}, \Vec{\sigma} \rangle 
 =\langle \Vec{x} - \Vec{a}, \Vec{\sigma} \rangle + \langle \beta^{m+1} \Vec{a}, \Vec{\sigma} \rangle \leq \beta^{m+1}$
 holds for any $\Vec{\sigma} \in \{-1,1\}^n$.  
 This means that if $\Vec{x} \in Q_k \cap Q_{m+1}$ then 
 $\langle \Vec{x} - \Vec{a}, \Vec{\sigma} \rangle 
  \leq \min \{\beta^{k}(1 - \langle \Vec{a}, \Vec{\sigma} \rangle), \beta^{m+1} (1- \langle \Vec{a}, \Vec{\sigma} \rangle) \} $. 
 It is not difficult to see that 
 $\min \{\beta^{k}(1 - \langle \Vec{a}, \Vec{\sigma} \rangle), \beta^{m+1} (1- \langle \Vec{a}, \Vec{\sigma} \rangle) \} 
  \leq \beta^{m}(1 - \langle \Vec{a}, \Vec{\sigma} \rangle)$. 
 Thus we obtain the claim.  
\end{proof}

\begin{proof}[Proof of Lemma~\ref{lem:vol-union-gem-seriesA}]
 The claim is trivial when $m=1$. 
 Inductively assuming that the claim holds in case of~$m$, 
  we obtain that 
\begin{align*}
 \bigcup_{k=0}^{m+1} Q_k 
 &=  
  \bigcup_{k=0}^{m} Q_k \,\cup\, Q_{m+1} \\
 &= 
  \left(\left(\dot\bigcup_{k=0}^{m-1} Q_k \setminus Q_{k+1}\right) \,\dot\cup\, Q_m\right) \,\cup\, Q_{m+1} 
  && (\mbox{Induction hypo.}) \\
 &= 
  \left(\dot\bigcup_{k=0}^{m-1} Q_k \setminus Q_{k+1}\right) \,\dot\cup\, (Q_m \,\cup\, Q_{m+1}) 
  && (\mbox{by Claim~\ref{claim}}) \\
 &= 
  \left(\dot\bigcup_{k=0}^{m-1} Q_k \setminus Q_{k+1}\right) \,\dot\cup\, 
  \left((Q_m \setminus Q_{m+1})  \,\dot\cup\, Q_{m+1} \right) \\
 &=
  \left(\dot\bigcup_{k=0}^m Q_k \setminus Q_{k+1}\right) \,\dot\cup\, Q_m 
\end{align*}
 which is the claim in case of $m+1$. 
\end{proof}

 The second step of the proof of Lemma~\ref{lem:vol-union-gem-seriesB} is the following lemma. 
\begin{lemma}\label{lem:vol-union-gem-seriesB}
\begin{eqnarray*}
 \Vol(Q_k \setminus Q_{k+1}) = \beta^{nk}\Vol(Q_0 \setminus Q_1). 
\end{eqnarray*}
\end{lemma}
\begin{proof}
 It is easy to see that 
 $\Vol(Q_k)/\Vol(Q_0) = \Vol(C((1-\beta^k) \Vec{a},\beta^k))/\Vol(C(\Vec{0},1)) = \beta^{nk}$ holds, 
 $\Vol(Q_{k+1})/\Vol(Q_1) = \beta^{nk}$ as well. 
 Using the bijective map $\eta_k \colon Q_k \to Q_0$ defined by~\eqref{eq:eta_k} in Lemma~\ref{lemma:subset}, 
  it is also not difficult to see that 
  $\Vol(Q_k \cap Q_{k+1}) / \Vol(Q_0 \cap Q_1) = \beta^{nk}$. 
 Considering the inclusion-exclusion, we obtain the claim. 
\end{proof}
Now, we are ready to prove Lemma~\ref{lem:vol-union-gem-series}. 
\begin{proof}[Proof of Lemma~\ref{lem:vol-union-gem-series}]
\begin{align*}
 \Vol\left(\bigcup_{k=0}^{\infty} Q_k\right)
 &= 
  \Vol\left(\left(\dot\bigcup_{k=0}^{\infty} Q_k \setminus Q_{k+1}\right) \,\dot\cup\, Q_{\infty}\right) 
 &(\mbox{by Lemma~\ref{lem:vol-union-gem-seriesA}}) \\
 &= \sum_{k=0}^{\infty} \Vol (Q_k \setminus Q_{k+1}) + \Vol(Q_{\infty}) \\
 &= \sum_{k=0}^{\infty} \beta^{nk}\Vol(Q_0 \setminus Q_1) 
 &(\mbox{by Lemma~\ref{lem:vol-union-gem-seriesB}}) \\
 &= \frac{1}{1-\beta^n} \Vol(Q_0 \setminus Q_1) \\
 &= \frac{1}{1-\beta^n} \left(\frac{2^n}{n!} -\Vol(Q_1 \cap Q_0)\right)
\end{align*}
\end{proof}

\subsubsection{The volume of $\bigcup_{k=0}^{\infty} Q_k$}\label{sec:red-two-ballsB}
 A reader who are familiar with approximation 
  may worry about the subtraction 
   $\frac{2^n}{n!} -\Vol(Q_0 \cap Q_1)$ in Lemma~\ref{lem:vol-union-gem-series}. 
 Section~\ref{sec:red-two-ballsB} claims the following. 
\begin{lemma}\label{lem:intersection-small}
 When $1-\beta \geq \dfrac{c_2\epsilon}{n\|\Vec{a}\|_1}$ where $0 < c_2 \epsilon < 1$, 
\begin{eqnarray*}
 \Vol(Q_0 \cap Q_1)  \leq \frac{1}{ 1+\dfrac{c_2 \epsilon}{2n} } \frac{2^n}{n!}. 
\end{eqnarray*}
\end{lemma}
 Intuitively, Lemma~\ref{lem:intersection-small} implies that 
   $\frac{2^n}{n!} -\Vol(Q_0 \cap Q_1)$ is large enough, and 
  an approximation of $ \Vol(Q_0 \cap Q_1)$ provides 
  a good approximation of $\Vol (\bigcup_{k=0}^{\infty} Q_k)$, and hence $\Vol(P_{\Vec{a}})$. 
 A detailed argument on our FPTAS of $\Vol(P_{\Vec{a}})$ will be described in Section~\ref{sec:algorithm-P_a}.

 As a preliminary of a proof of Lemma~\ref{lem:intersection-small}, 
  we give Lemmas~\ref{lem:intersection-smallA} and \ref{lem:intersection-smallB}. 
\begin{lemma}\label{lem:intersection-smallA}
 Let $\Vec{c} \in \mathbb{R}^n_{\geq 0}$ and $\Vec{c}' \in \mathbb{R}^n_{\geq 0}$ be given by 
\begin{eqnarray*}
 \Vec{c} &=& (c_1,c_2, \ldots, c_{k-1}, c_k,0,\ldots,0) \\
 \Vec{c}' &=& (c_1,c_2, \ldots, c_{k-1},0,0,\ldots,0) 
\end{eqnarray*}
 for some $k \in \{2,3,\ldots,n\}$, 
 i.e., $\Vec{c}'$ is given by replacing the $k$-th component of $\Vec{c}$ by $0$. 
 Then, 
\begin{eqnarray}
 \Vol (C(\Vec{0},1) \cap C(\Vec{c},r)) 
\leq 
 \Vol (C(\Vec{0},1) \cap C(\Vec{c}',r)). 
\label{eq:intersection^smallA}
\end{eqnarray}
\end{lemma}
\begin{proof}
 For any $\Vec{x}= (x_1,\ldots,x_n) \in C(\Vec{c},r)$, 
  we define a map $h \colon C(\Vec{c},r) \to \mathbb{R}^n$ 
  such that $\Vec{x}' = h(\Vec{x})$ satisfies 
\begin{eqnarray*}
 x'_i = \begin{cases}
 c_i-x_i & (\mbox{for }i=k) \\
 x_i & (\mbox{otherwise}) 
 \end{cases}
\end{eqnarray*}
 (see Figure~\ref{fig:h(x)}). 
\begin{figure}
\begin{center}
\includegraphics[bb= 20 40 440 440, scale=0.5, clip]{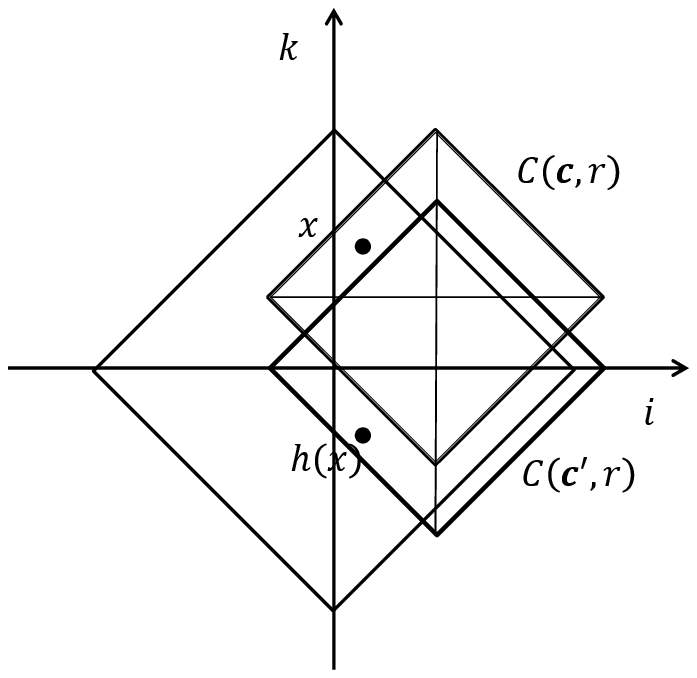}
\end{center}
\caption{$\Vec{x}$ and $h(\Vec{x})$. }\label{fig:h(x)}
\end{figure}
 Notice that 
  $h$ is a map from $C(\Vec{c},r)$ to $C(\Vec{c}',r)$ in fact, and 
  it is  bijective and measure preserving.  
 Now, suppose that $\Vec{x} \in C(\Vec{c},r)$ 
  satisfies both $\Vec{x}  \in  C(\Vec{0},1)$ and  $\Vec{x} \not\in C(\Vec{c}',r)$, 
 i.e., 
\begin{eqnarray}
 \textstyle \sum_{i=1}^k |x_i-c_i| + \sum_{i=k+1}^n |x_i| \leq r, && \label{20160403p11}\\
 \textstyle \sum_{i=1}^n |x_i| \leq 1, &&  \mbox{ and } \label{20160403p12} \\
 \textstyle \sum_{i=1}^{k-1} |x_i-c_i| + \sum_{i=k}^n |x_i| > r \label{20160403p13} && 
\end{eqnarray} 
 hold. 
 Then, we claim that 
  $\Vec{x}' = h(\Vec{x}) \in C(\Vec{0},1)$ and $\Vec{x}' \not\in C(\Vec{c},r)$. 
 This implies \eqref{eq:intersection^smallA} since $h$ is measure preserving. 
 Now we show the claim. 
 For convenience let 
\begin{eqnarray*}
 D := \textstyle \sum_{i=1}^{k-1} |x_i-c_i| + \sum_{i=k+1}^n |x_i| 
\end{eqnarray*} 
 then 
  \eqref{20160403p11} implies $D+|x_k-c_k| \leq r$ and 
  \eqref{20160403p13} implies $D+|x_k| > r$.  
 As a consequence, 
  we obtain that 
\begin{eqnarray}
 |x_k-c_k| < |x_k|. \label{20160403p14}
\end{eqnarray} 
 We also remark that $|x'_k| = |x_k-c_k|$ by the definition of $h$. 
 Then 
\begin{align*}
 \|\Vec{x}' \|_1 
 &= \textstyle \sum_{i=1}^n |x'_i| = \left(\sum_{i=1}^n |x_i|\right) -|x_k|+|x_k-c_k|  \\
 & < \textstyle \sum_{i=1}^n |x_i| 
 && (\mbox{by \eqref{20160403p14}}) \\
 & \leq 1 
 && (\mbox{by \eqref{20160403p12}}) 
\end{align*} 
  and $\Vec{x}' \in C(\Vec{0},1)$. 
 Similarly, 
\begin{align*}
 \|\Vec{x}' - \Vec{c} \|_1 
 &= 
  \textstyle \sum_{i=1}^{k-1} |x'_i-c_i| + |x'_k| + \sum_{i=k+1}^n |x'_i| \\
 &= 
  \textstyle \sum_{i=1}^{k-1} |x_i-c_i| + |x_k-c_k| + \sum_{i=k+1}^n |x_i| \\
 &> 
  \textstyle \sum_{i=1}^{k-1} |x_i-c_i| + |x_k| + \sum_{i=k+1}^n |x_i| 
 && (\mbox{by \eqref{20160403p14}}) \\
 &\geq   r  
 && (\mbox{by \eqref{20160403p13}}) 
\end{align*} 
 and $\Vec{x}' \not\in C(\Vec{c},r)$. 
 We obtain the claim. 
\end{proof}
 We remark that the volume of the intersection is not 
  monotone decreasing with respect to the $L_1$ distance between centers, in general. 
 Iteratively applying Lemma~\ref{lem:intersection-smallA}, 
  we see the following. 
\begin{corollary}\label{cor:intersection-smallA}
 Let $\Vec{c} \in \mathbb{R}^n_{\geq 0}$ and $\Vec{c}' \in \mathbb{R}^n_{\geq 0}$ be given by 
\begin{eqnarray*}
 \Vec{c} &=& (c_1,c_2,\ldots, c_n) \quad \mbox{and}\\
 \Vec{c}' &=& (c_1,0,\ldots,0), 
\end{eqnarray*}
  i.e., $\Vec{c}'$ is given by replacing each component, except for the first component, of $\Vec{c}$ by 0. 
 Then, 
\begin{eqnarray*}
 \Vol (C(\Vec{0},1) \cap C(\Vec{c},r)) 
\leq 
 \Vol (C(\Vec{0},1) \cap C(\Vec{c}',r)). 
\end{eqnarray*}
\end{corollary}

 Next, we show the following. 
\begin{lemma}\label{lem:intersection-smallB}
 Suppose that $r$ and $c$ satisfies $0<r<1$ and $0<c<1+r$. 
 Then, 
\begin{eqnarray}
 C(\Vec{0},1) \cap C((c,0,\ldots,0),r) 
 &\subseteq& C\left(\left(\frac{1+(c-r)}{2},0,\ldots,0\right),\frac{1-(c-r)}{2}\right) 
\label{eq:intersection-smallB}
\end{eqnarray} 
holds. 
\end{lemma}
 Remark that \eqref{eq:intersection-smallB} in fact holds by equality, but we here prove only $\subseteq$, which we will use. 
\begin{proof}
 Suppose that $\Vec{x} \in  C(\Vec{0},1) \cap C((c,0,\ldots,0),r)$, 
 i.e., 
\begin{eqnarray}
  \textstyle \sum_{i=1}^n |x_i| \leq 1 &&\mbox{ and }
\label{20160403c}\\
 \textstyle |x_1-c| + \sum_{i=2}^n |x_i| \leq r&&
\label{20160403d}
\end{eqnarray} 
 holds. 
We consider two cases. 

\noindent
(i) In case that $x_1 \geq \frac{1+(c-r)}{2}$, 
\begin{eqnarray*}
\left|x_1-\frac{1+(c-r)}{2}\right|+\sum_{i=2}^n |x_i| 
&=& x_1-\frac{1+(c-r)}{2}+\sum_{i=2}^n |x_i| \\
&=& \sum_{i=1}^n |x_i|-\frac{1+(c-r)}{2} \\
&\leq& 1-\frac{1+(c-r)}{2}  \hspace{3em}(\mbox{by \eqref{20160403c}})\\
&=& \frac{1-(c-r)}{2} 
\end{eqnarray*} 
and we see that $\Vec{x} \in C\left(\left(\frac{1+(c-r)}{2},0,\ldots,0\right),\frac{1-(c-r)}{2}\right) $. 

\noindent
(ii) In case that $x_1 < \frac{1+(c-r)}{2}$, 
\begin{eqnarray*}
\left|x_1-\frac{1+(c-r)}{2}\right|+\sum_{i=2}^n |x_i| 
&=& \frac{1+(c-r)}{2}-x_1+\sum_{i=2}^n |x_i| \\
&=& \frac{1+(c-r)}{2}-c+(c-x_1)+\sum_{i=2}^n |x_i| \\
&\leq& \frac{1+(c-r)}{2}-c+|x_1-c|+\sum_{i=2}^n |x_i| \\
&\leq& \frac{1+(c-r)}{2}-c+r \hspace{3em}(\mbox{by \eqref{20160403d}})\\
&=& \frac{1-(c-r)}{2} 
\end{eqnarray*} 
and we see that $\Vec{x} \in C\left(\left(\frac{1+(c-r)}{2},0,\ldots,0\right),\frac{1-(c-r)}{2}\right) $. 
\end{proof}

Now we prove Lemma~\ref{lem:intersection-small}.  
\begin{proof}[Proof of Lemma~\ref{lem:intersection-small}]
 Recall that $Q_1 = C((1-\beta)\Vec{a},\beta)$. 
 By Corollary~\ref{cor:intersection-smallA} and Lemma~\ref{lem:intersection-smallB}, 
\begin{eqnarray*}
 \Vol(Q_0 \cap Q_1)  
 &\leq& \Vol\left(C(\Vec{0},1) \cap C(((1-\beta)a_1,0,\ldots,0),\beta)\right) 
  \hspace{2em} (\mbox{by Corollary~\ref{cor:intersection-smallA}})\\
 &\leq& 
 \Vol\left(C\left(\left(\frac{1+((1-\beta)a_1-\beta)}{2},0,\ldots,0\right),\frac{1-((1-\beta)a_1-\beta)}{2}\right)\right) \\
 && \hspace{20em} (\mbox{by Lemma~\ref{lem:intersection-smallB}})\\
 &=& \left(\frac{1-((1-\beta)a_1-\beta)}{2}\right)^n \Vol(C(\Vec{0},1))\\
 &=& \left(\frac{2-(1-\beta)(a_1+1)}{2}\right)^n \Vol(C(\Vec{0},1))\\
 &=& \left(1-(1-\beta)\frac{a_1+1}{2}\right)^n \Vol(C(\Vec{0},1))\\
 &\leq& \left(1-\dfrac{c_2 \epsilon}{ n\|\Vec{a}\|_1}\frac{a_1+1}{2}\right)^n \Vol(C(\Vec{0},1))  
 \hspace{2em}\left(\mbox{since }1-\beta \geq \frac{c_2\epsilon}{n\|\Vec{a}\|_1} \mbox{ (hypo.)}\right)
\\
 &\leq& \left(1-\dfrac{c_2 \epsilon}{2n^2} \right)^n \Vol(C(\Vec{0},1))
 \hspace{6.3em}\left(\mbox{since }\frac{a_1+1}{2\|\Vec{a}\|_1}\geq \frac{a_1+1}{2na_1} \geq \frac{1}{2n}\right)
\\
 &\leq& \frac{1}{\left(1+\dfrac{c_2 \epsilon}{2n^2}\right)^n} \Vol(C(\Vec{0},1))\\
 &\leq& \frac{1}{1+\dfrac{c_2 \epsilon}{2n}} \Vol(C(\Vec{0},1))
\end{eqnarray*}
and we obtain the claim. 
\end{proof}

\subsection{Approximation algorithm and analysis}\label{sec:algorithm-P_a}
 Based on 
  Lemma~\ref{lem:geometric-series} in Section~\ref{sec:geometric-series} and 
  Lemma~\ref{lem:vol-union-gem-series} in Section~\ref{sec:red-two-balls}, 
  we give an FPTAS for $\Vol(P_{\Vec{a}})$ where we assume an algorithm to approximate $\Vol(Q_0 \cap Q_1)$. 
 For convenience of arguments, 
  we assume $ 0 < \epsilon \leq 1/2$, 
  but it is clearly not essential\footnote{
   For $\epsilon > 1/2$, use Algorithm~\ref{algorithm:GDKP} with $\epsilon = 1/2$. 
 }. 
\begin{algorithm}[ ($1 \pm \epsilon$)-approximation ($ 0 < \epsilon \leq 1/2$)]\label{algorithm:GDKP}{\normalfont \ \\
Input: $\Vec{a}\in \Integer_+^n$; \\
1. Set parameter $\beta:=1-\dfrac{\epsilon}{2n\|\Vec{a}\|_1}$; 
2. Approximate $I \defeq \Vol(C(\Vec{0},1) \cap C((1-\beta)\Vec{a},\beta))$ 
   by $Z$ such that 
\begin{eqnarray*}
  I \le Z \le \left(1+\dfrac{\epsilon^2}{4n}\right)I; 
\end{eqnarray*}
3. Output 
\begin{eqnarray*}
 \widehat{V} = \frac{1+\epsilon}{1-\beta^n}\left(\frac{2^n}{n!}-Z\right).
\end{eqnarray*}
}\end{algorithm}

\begin{lemma}\label{lem:GDKP}
 The output $\widehat{V}$ of Algorithm~\ref{algorithm:GDKP} satisfies 
\begin{eqnarray*}
 \left(1-\epsilon\right)  \Vol(P_{\Vec{a}}) 
 \leq \widehat{V} 
 \leq (1+\epsilon)\Vol(P_{\Vec{a}}). 
\end{eqnarray*}
\end{lemma}
 Before proving Lemma~\ref{lem:GDKP}, 
  we check the time complexity of Algorithm~\ref{algorithm:GDKP}. 
 In Section~\ref{sec:fptas-two-cp}, we will give an FPTAS for $\Vol(Q_0 \cap Q_1)$. 
 Theorem~\ref{th:CPI} appearing there 
  implies that the time complexity of Step 2 of Algorithm~\ref{algorithm:GDKP} 
  is 
  $\Order(n^7 (n/\epsilon^2)^3)=\Order(n^{10} \epsilon^{-6})$. 
 Thus, we obtain Theorem~\ref{th:main} by Lemma~\ref{lem:GDKP}. 

 As a preliminary of Lemma~\ref{lem:GDKP}, we show the following. 
\begin{lemma}\label{lem:GDKP1}
 Suppose that $1-\beta \geq \dfrac{c_2 \epsilon}{n\|\Vec{a}\|_1}$ holds where $0 < c_2 \epsilon < 1 $. 
 If we have an approximation $Z$ of $\Vol(Q_1 \cap Q_0)$ satisfying 
\begin{eqnarray}
 \Vol(Q_0 \cap Q_1) \le Z \le \left(1+\frac{c_2\epsilon^2}{2n}\right)\Vol(Q_0 \cap Q_1)
\label{eq:GDKP-assumption}
\end{eqnarray}
 then $\Vol(Q_0) - Z = \dfrac{2^n}{n!}-Z$ satisfies that 
\begin{eqnarray}
 (1-\epsilon)\!\!\left(\frac{2^n}{n!} -\Vol(Q_1 \cap Q_0)\right) 
 \leq \left(\frac{2^n}{n!} -Z\right)
 \leq \left(\frac{2^n}{n!} -\Vol(Q_1 \cap Q_0)\right). 
\label{eq:GDKP-main}
\end{eqnarray}
\end{lemma} 
\begin{proof}
 The second inequality of \eqref{eq:GDKP-main} is easy from the assumption~\eqref{eq:GDKP-assumption}, 
  such that 
\begin{eqnarray*}
 \frac{2^n}{n!} - Z 
  &\leq& \frac{2^n}{n!} - \Vol(Q_1 \cap Q_0)  
\end{eqnarray*} 
  holds.  
 For the first inequality of \eqref{eq:GDKP-main}, \eqref{eq:GDKP-assumption}
\begin{align}
 \frac{2^n}{n!} - Z 
  &\geq 
   \frac{2^n}{n!} - \left(1+\frac{c_2\epsilon^2}{2n}\right)\Vol(Q_1 \cap Q_0) 
  && (\mbox{by \eqref{eq:GDKP-assumption}}) \nonumber\\
  &=  
   \left(\frac{2^n}{n!} - \Vol(Q_1 \cap Q_0)\right) - \frac{c_2\epsilon^2}{2n} \Vol(Q_1 \cap Q_0) \nonumber\\
  &=  
   \left(\frac{2^n}{n!} - \Vol(Q_1 \cap Q_0)\right)  
   \left(1 - \frac{c_2\epsilon^2}{2n}  \frac{\Vol(Q_1 \cap Q_0)}{\dfrac{2^n}{n!} - \Vol(Q_1 \cap Q_0)} \right) 
\label{eq:tmp160219a}
\end{align} 
  holds. 
 Since the hypothesis $1-\beta \geq \dfrac{c_2\epsilon}{n\|\Vec{a}\|_1}$, 
 Lemma~\ref{lem:intersection-small} implies that 
\begin{eqnarray*}
\Vol(Q_0 \cap Q_1)  \leq \frac{1}{1+\dfrac{c_2 \epsilon}{2n}} \frac{2^n}{n!} 
\end{eqnarray*}
 and hence
\begin{eqnarray*}
 \frac{\Vol(Q_0 \cap Q_1)}{\dfrac{2^n}{n!}-\Vol(Q_0 \cap Q_1)} 
 &=&
 \frac{1}{\dfrac{\frac{2^n}{n!}}{\Vol(Q_0 \cap Q_1)}-1} \\
 &\leq& 
  \frac{1}{1 + \dfrac{c_2 \epsilon}{2n} -1} \\
 &=& 
  \dfrac{2n}{c_2 \epsilon}
\end{eqnarray*}
  holds. 
 Thus,  
\begin{eqnarray*}
\eqref{eq:tmp160219a}
  & \geq & 
   \left(\frac{2^n}{n!} - \Vol(Q_1 \cap Q_0)\right)  
   \left(1 - \frac{c_2\epsilon^2}{2n}  \frac{2n}{c_2 \epsilon} \right) \\
  &=& 
   \left(\frac{2^n}{n!} - \Vol(Q_1 \cap Q_0)\right) 
   (1 - \epsilon) 
\end{eqnarray*} 
 holds, and we obtain the claim.  
\end{proof}

\begin{corollary}\label{cor:error}
 Let $1-\beta = \dfrac{c\epsilon}{n\|\Vec{a}\|_1}$ with $\dfrac{1}{4} \leq c \leq \dfrac{1}{2}$. 
 If we have an approximation $Z$ of $\Vol(Q_1 \cap Q_0)$ satisfying 
\begin{eqnarray}
 \Vol(Q_0 \cap Q_1) \le Z \le \left(1+\frac{c\epsilon^2}{2n}\right)\Vol(Q_0 \cap Q_1)
\label{cor1:hypo}
\end{eqnarray}
 then 
\begin{eqnarray*}
 \widehat{V}:=\frac{1+\epsilon}{1-\beta^n} \left(\frac{2^n}{n!} -Z\right)
\end{eqnarray*}
 satisfies 
\begin{eqnarray*}
 \left(1-\epsilon\right)  \Vol(P_{\Vec{a}}) 
 \leq \widehat{V} 
 \leq (1+\epsilon)\Vol(P_{\Vec{a}}). 
\end{eqnarray*}
\end{corollary}
\begin{proof}
 Recall Lemma~\ref{lem:vol-union-gem-series}, that is 
\begin{eqnarray*}
\Vol\left(\bigcup_{k=0}^\infty Q_k\right)
 = \frac{1}{1-\beta^n} \left(\frac{2^n}{n!} -\Vol(Q_1 \cap Q_0)\right). 
\end{eqnarray*}
 By Lemma~\ref{lem:geometric-series}, 
  the hypothesis $1-\beta \leq \dfrac{\frac{1}{2}\epsilon}{n\|\Vec{a}\|_1}$ implies that 
\begin{eqnarray*}
 \left(1-\frac{1}{2}\epsilon \right)\Vol(P_{\Vec{a}}) 
  \leq \Vol\left(\bigcup_{k=0}^\infty Q_k\right)
  \leq \Vol(P_{\Vec{a}}) 
\end{eqnarray*}
 holds. 
 Thus, \eqref{cor1:hypo} implies that 
\begin{eqnarray*}
 \widehat{V}
  &=& \frac{1+\epsilon}{1-\beta^n} \left(\frac{2^n}{n!} -Z\right) \\
  &\leq & \frac{1+\epsilon}{1-\beta^n} \left(\frac{2^n}{n!} -\Vol(Q_1 \cap Q_0)\right) \\
  &\leq & (1+\epsilon)\Vol(P_{\Vec{a}})
\end{eqnarray*}
 and we obtain the upper bound. 
 Similarly, 
\begin{eqnarray*}
 \widehat{V}
  &=& \frac{1+\epsilon}{1-\beta^n} \left(\frac{2^n}{n!} -Z\right) \\
  &\geq & \frac{1+\epsilon}{1-\beta^n} (1-\epsilon)\left(\frac{2^n}{n!} -\Vol(Q_1 \cap Q_0)\right) \\
  &\geq & (1-\epsilon^2)\frac{1}{1-\beta^n} \left(\frac{2^n}{n!} -\Vol(Q_1 \cap Q_0)\right) \\
  &\geq & (1-\epsilon^2)\left(1-\frac{\epsilon}{2}\right)\Vol(P_{\Vec{a}}) \\
  &\geq & (1-\epsilon)\Vol(P_{\Vec{a}})  \hspace{2em}\left(\mbox{by assumption } \epsilon \leq \frac{1}{2}\right)
\end{eqnarray*}
  and we obtain the claim. 
\end{proof}
 Now,  Lemma~\ref{lem:GDKP} is immediate from Corollary~\ref{cor:error}.

\section{The Volume of the Intersection of Two Cross-polytopes}\label{sec:fptas-two-cp}
 This section gives an FPTAS for the volume of 
the intersection of two cross-polytopes in the $n$-dimensional space. 
 Without loss of generality\footnote{ 
 Remark that 
$\Vol(C(\Vec{c},r)\cap C(\Vec{c}',r')) =
 r^n\Vol\left(C(\Vec{0},1) \cap C\!\left(\tfrac{(\Vec{c}-\Vec{c}')^+}{r}, \tfrac{r'}{r} \right)\right)$
 holds for any $\Vec{c}, \Vec{c}' \in \mathbb{R}^n$ and $r,r' \in \mathbb{R}_{>0}$, 
 where $(\Vec{c} - \Vec{c}')^+ = (|c_1-c_1'|,|c_2-c_2'|,\ldots,|c_n-c_n'| )$. }, 
  we are concerned with $\Vol(C(\Vec{0},1) \cap C(\Vec{c},r))$ for $\Vec{c} \geq \Vec{0}$ and $r$ $(0 < r \le 1)$. 
 This section establishes the following. 
\begin{theorem}\label{th:CPI}
 For any $\delta$ $(0<\delta<1)$, 
  there exists a {\em deterministic} algorithm 
  which outputs a value $Z$ satisfying 
    $\Vol(C(\Vec{0},1)\cap C(\Vec{c},r))\le Z \le (1+\delta)\Vol(C(\Vec{0},1)\cap C(\Vec{c},r))$
  for any input $\Vec{c}\geq\Vec{0}$ and $r$ $(0 < r \leq 1)$ satisfying $\|\Vec{c}\|_1 \leq r$, and  
  runs in $O(n^7 \delta^{-3})$ time. 
\end{theorem}
 
 The assumption that $\|\Vec{c}\|_1 \leq r$ implies 
   both centers $\Vec{0}$ and $\Vec{c}$ are contained in the intersection $C(\Vec{0},1) \cap C(\Vec{c},r)$. 
 Note that the assumption does not harm to our main goal Theorem~\ref{th:main} 
   (recall Algorithm~\ref{algorithm:GDKP} in Section~\ref{sec:algorithm-P_a}). 
 We show in Section~\ref{sec:hardness} that 
  $\Vol(C(\Vec{0},1) \cap C(\Vec{c},r))$ remains $\#P$-hard even on the assumption. 
 We will use the assumption in the proof of Lemma~\ref{lem:cone}.

\subsection{Preliminary: convolution for the volume}\label{sec:convol}
 As a preliminary step, 
 Section~\ref{sec:convol} gives a convolution which provides $\Vol(C(\Vec{0},1)\cap C(\Vec{c},r))$. 
 Let $\Psi_0 \colon \mathbb{R}^2 \to \mathbb{R}$ be given by 
  $\Psi_0(u,v)=1$ if $u \geq 0$ and $v \geq 0$, 
  otherwise $\Psi_0(u,v)=0$. 
 Inductively, we define $\Psi_i \colon \mathbb{R}^2 \to \mathbb{R}$ for $i=1,2,\ldots,n$ by 
\begin{eqnarray}
\Psi_i(u,v)
 &\defeq & 
 \int_{-1}^{1} \Psi_{i-1}(u-|s|, v-|s-c_i|) \dd s 
\label{eq:def-Psi}
\end{eqnarray}
  for $u,v \in \mathbb{R}$. 
 We remark that $\Psi_i(u,v) = 0$ holds if $u \leq 0$ or $v \leq 0$, for any $i=1,2,\ldots,n$ by the definition. 
\begin{lemma}\label{lem:Psi_n(1,r)}
\begin{eqnarray*}
 \Psi_n(1,r)
 =
 {\rm Vol}(C(\Vec{0},1)\cap C(\Vec{c}, r)). 
\end{eqnarray*}
\end{lemma}

 To prove Lemma~\ref{lem:Psi_n(1,r)}, 
  it might be helpful to introduce a probability space. 
 Let $\Vec{X} = (X_1,\ldots,X_n)$ be a uniform random variable over $[-1,1]^n$, 
  i.e., $X_i$ ($i=1,\ldots,n$) are (mutually) independent. 
 Then, 
\begin{eqnarray}
 \Prob\left[ X \in C(\Vec{0},1)\cap C(\Vec{c},r) \right] 
  = \frac{\Vol(C(\Vec{0},1)\cap C(\Vec{c},r))}{\Vol([-1,1]^n)}
  = \frac{1}{2^n}\Vol(C(\Vec{0},1)\cap C(\Vec{c},r))
\label{eq:prob}
\end{eqnarray} 
 holds. 

\begin{lemma}\label{lem:Psi_n(1,r)-lem} 
For any $i=1,2,\ldots,n$, 
\begin{eqnarray*} 
\dfrac{1}{2^i} \Psi_i(u,v) 
 = 
  \Pr\left[\left(\sum_{j=1}^{i} |X_{j}| \le u \right) \wedge 
           \left(\sum_{j=1}^{i}|X_{j}-c_{j}|\le v \right)\right]
\end{eqnarray*}
for any $u,v \in \mathbb{R}$. 
\end{lemma}
\begin{proof}
 First, we prove the claim for $i=1$. Considering that $\Psi_0(u, v)$ is an indicator function, 
\begin{eqnarray*} 
 \Psi_1(u,v) 
 &=& 
 \int_{-1}^{1} \Psi_0(u-|s|, v-|s-c_i|) \dd s \\
 &=& 
 \left|\{s \in [-1,1] \mid (u-|s| \geq 0) \wedge (v-|s-c_i|\geq 0)\}\right| \\
 &=& 
 2\frac{\left|\{s \in [-1,1] \mid (u-|s| \geq 0) \wedge (v-|s-c_i|\geq 0)\}\right|}{|[-1,1]|} \\
 &=& 
  2\Pr\left[\left(0 \le u-|X_1| \right) \wedge 
           \left(0 \le v-|X_1-c_1| \right)\right] \\
 &=& 
  2\Pr\left[\left(|X_1| \le u \right) \wedge 
           \left(|X_1-c_1|\le v \right)\right] 
\end{eqnarray*}
 and we obtain the claim in the case. 

 Inductively assuming that the claim for $i$, 
  we show that the claim for $i+1$. 
 Let $f$ denote the uniform density over $[-1,1]$. Then, 
\begin{align*}
\lefteqn{ \textstyle
  \Pr\left[\left(\sum_{j=1}^{i+1} \left| X_{j} \right| \le u \right) \wedge 
           \left(\sum_{j=1}^{i+1} \left| X_{j}- c_{j} \right|\le v \right)\right] 
}\\
 &= 
 \int_{-\infty}^{\infty} 
 {\textstyle
  \Pr\left[\left(\sum_{j=1}^{i+1} \left| X_{j} \right| \le u \right) \wedge 
           \left(\sum_{j=1}^{i+1} \left| X_{j}- c_{j} \right|\le v \right) \mid X_{i+1}=s \right] f(s) {\rm d}s }
\\
 &= 
 \int_{-1}^{1} 
 {\textstyle
  \Pr\left[\left(\sum_{j=1}^{i+1} \left| X_{j} \right| \le u \right) \wedge 
           \left(\sum_{j=1}^{i+1} \left| X_{j}- c_{j} \right|\le v \right) \mid X_{i+1}=s \right] \dfrac{1}{2} {\rm d}s }
\\
 &= \dfrac{1}{2} 
 \int_{-1}^{1} 
 {\textstyle
  \Pr\left[\left(\sum_{j=1}^{i} \left| X_{j} \right| + |s|  \le u  \right) \wedge 
           \left(\sum_{j=1}^{i} \left| X_{j}- c_{j} \right| + \left| s- c_{i+1} \right| \le v \right) \right] {\rm d}s }
\\
 &= \dfrac{1}{2} 
 \int_{-1}^{1} 
 {\textstyle
  \Pr\left[\left(\sum_{j=1}^{i} \left| X_{j} \right|  \le u-|s|  \right) \wedge 
           \left(\sum_{j=1}^{i} \left| X_{j}- c_{j} \right| \le v - \left| s- c_{i+1} \right|  \right) \right] {\rm d}s }
\\
 &= \dfrac{1}{2} \int_{-1}^{1} \dfrac{1}{2^{i}} \Psi_{i}(u-|s|,v-|s- c_{i+1}|) {\rm d}s \\
 &= \dfrac{1}{2^{i+1}} \Psi_{i+1}(u,v)
\end{align*}
 and we obtain the claim. 
\end{proof}
 Now, Lemma~\ref{lem:Psi_n(1,r)} is easy from Lemma~\ref{lem:Psi_n(1,r)-lem} and \eqref{eq:prob}.

\subsection{Idea for approximation}\label{sec:G-idea}
 Our FPTAS is based on an approximation of $\Psi_i(u,v)$. 
 Let $G_0(u,v) = \Psi_0(u,v)$ for any $u,v \in \mathbb{R}$, 
  i.e., $G_0(u,v)=1$ if $u \geq 0$ and $v \geq 0$, 
  otherwise $G_0(u,v)=0$.
 Inductively assuming $G_{i-1}(u,v)$, we define 
\begin{align}
 \overline{G}_i(u,v) 
 &\defeq 
 \int_{-1}^1 G_{i-1}(u-|s|,v-|s-c_i|) \dd s 
 \label{form:Gbar1}
\end{align}
  for $u,v \in \mathbb{R}$, for convenience.  
 Then, let $G_i(u,v)$ be a staircase approximation of $\overline{G}_i(u,v)$, given by 
\begin{align}
 G_i(u,v)
  \defeq \begin{cases}
    \overline{G}_i \left(\tfrac{1}{M}k, \tfrac{r}{M}\ell \right) 
 & \left( \begin{array}{lll}\mbox{if}
   &\tfrac{1}{M}(k-1) < u \le \tfrac{1}{M}k       & (k =1,2,\ldots), \mbox{ and } \\ 
   &\tfrac{r}{M}(\ell-1) < v \le \tfrac{r}{M}\ell & (\ell = 1,2,\ldots). \end{array}\right) \\
    0 & (\text{otherwise})
    \end{cases} \label{form:G}
\end{align}
 for any $u,v \in \mathbb{R}$. 
 Thus, we remark that 
\begin{eqnarray}
 G_i(u,v) = G_i \!\left(\tfrac{1}{M}\lceil Mu \rceil, \tfrac{r}{M} \!\left\lceil \tfrac{M}{r} v \right\rceil \right) 
\label{eq:remark-G}
\end{eqnarray}
holds for any $u,v \in \mathbb{R}$, by the definition. 
 Section~\ref{sec:G-algo} will show that $G_i(u,v)$ approximates $\Psi_i(u,v)$ well. 

 In the rest of Section~\ref{sec:G-idea}, 
  we briefly comment on the computation of $G_i$. 
 First, remark that \eqref{form:Gbar1} implies that 
  $\overline{G}_i(u,v)$ is computed only from $G_{i-1}(u',v')$ for $u' \leq u$ and $v' \leq v$, 
  i.e., we do not need to know  $G_{i-1}(u',v')$ for $u' > u$ or $v' > v$.  
 Second, remark \eqref{eq:remark-G} implies that 
  $G_i(u,v)$ for $u \leq 1$ and $v \leq r$  
   takes (at most) $(M+1)^2$ different values. 
 Precisely, 
let 
\begin{eqnarray*}
 \Gamma&\defeq& \left\{ \frac{1}{M} (k, r\ell) \mid k=0, 1, 2,\ldots, M,\  \ell= 0, 1, 2,\ldots, M \right\} \\
\end{eqnarray*}
 then $G_i(u,v)$ for $(u,v) \in \Gamma$ provides 
  all possible values of $G_i(u,v)$ for $u \leq 1$ and $v \leq r$, since \eqref{eq:remark-G}. 

 Then, we explain how to compute $G_i(u,v)$ for $(u,v) \in \Gamma$ from $G_{i-1}$. 
 For an arbitrary $(u,v) \in \Gamma$, let 
\begin{eqnarray*}
 S(u) 
  &\defeq& \left\{ s \in [-1,1] \mid u-|s| =\tfrac{1}{M}k  \ (k=0,1,2,\ldots,M) \right\} \\
  &=& \left\{ s \in [-1,1] \mid s = \pm (u-\tfrac{1}{M}k)  \ (k=0,1,2,\ldots,M) \right\}, 
\end{eqnarray*}
 let
\begin{eqnarray*}
 S_i(v) 
  &\defeq& \left\{ s \in [-1,1] \mid v-|s-c_i| = \tfrac{r}{M}\ell  \ (\ell=0,1,2,\ldots,M) \right\} \\
  &=& \left\{ s \in [-1,1] \mid s = c_i \pm (v - \tfrac{r}{M}\ell)  \ (\ell=0,1,2,\ldots,M) \right\},  
\end{eqnarray*}
  and let 
\begin{eqnarray*}
 T_i(u,v) &\defeq& S(u) \cup S_i(v) \cup \{-1,0,c_i,1\}
\end{eqnarray*}
 Suppose $t_0,t_1,\ldots,t_m$ be an ordering of all elements of $T_i(u,v)$ 
  such that $t_i \leq t_{i+1}$ for any $i=0,1,\ldots,m$, where $m=|T_i(u,v)|$. 
Then, we can compute $G_i(u,v)$ for any $(u,v) \in \Gamma$ by 
\begin{eqnarray}
 G_i(u,v) 
 &=& \overline{G}_i(u,v) \nonumber\\
 &=& \int_{-1}^1 G_{i-1}(u-|s|, v-|s-c_i|) \dd s \nonumber \\
 &=& \sum_{j=0}^{m-1} \int_{t_j}^{t_{j+1}} G_{i-1}(u-|s|,v-|s-c_i|) \dd s \nonumber \\
 &=& \sum_{j=0}^{m-1} \int_{t_j}^{t_{j+1}} 
     G_{i-1}\left(\tfrac{1}{M} \lceil M(u-|t_{j+1}|) \rceil, \tfrac{r}{M} \left\lceil \tfrac{M}{r}(v-|t_{j+1}-c_i|)\right\rceil \right) \dd s
 \qquad(\mbox{by \eqref{eq:remark-G}}) 
 \nonumber\\
 &=& \sum_{j=0}^{m-1} (t_{j+1}-t_j)
     G_{i-1}\left(\tfrac{1}{M} \lceil M(u-|t_{j+1}|) \rceil, \tfrac{r}{M} \left\lceil \tfrac{M}{r}(v-|t_{j+1}-c_i|)\right\rceil \right) 
\label{form:Gbar}
\end{eqnarray}
 where we remark again that the terms of \eqref{form:Gbar} consist of $G_{i-1}(u,v)$ for $(u,v) \in \Gamma$.

\subsection{Algorithm and analysis}\label{sec:G-algo}
 Based on the arguments in Section~\ref{sec:G-idea}, 
 our algorithm is described as follows. 
\begin{algorithm}[for $(1+\delta)$-approximation ($0<\delta \leq 1$)] 
\label{algorithm:CPI}\ \\
  Input: $\Vec{c}\in \mathbb{Q}_{\ge 0}^n$, $r \in \mathbb{Q}$ ($0 \leq r \le 1$); \\
  1. Set $M:= \lceil 4n^2\delta^{-1} \rceil$; \\
  2. Set $G_0(u,v):=1$ for $(u,v) \in \Gamma$, otherwise $G_0(u,v):=0$; \\
  3. For $i:=1,\ldots,n$, \\
  4. \hspace*{3mm} For $(u,v) \in \Gamma$, \\
  5. \hspace*{6mm} Compute $G_i(u,v)$ from $G_{i-1}$ by \eqref{form:Gbar}; \\
  6. Output $G_n(1, r)$. 
\end{algorithm}
\begin{lemma}\label{lem:G-time}
  The running time of Algorithm~\ref{algorithm:CPI} is $\Order(n^7\delta^{-3})$. 
\end{lemma}
\begin{proof}
 First, we are concerned with the running time of line 5. 
 The equation \eqref{form:Gbar} is a sum consisting of $m$ terms, 
  where clearly $m \leq 2M+4 =\Order(M)$. 
 We specially note that the ordering $t_0,t_1,t_2,\ldots,t_m$ of $T_i(u,v)$ is 
  obtained in $\Order(M)$ time, and hence line 5 runs in $\Order(M)$ time. 
 Since $|\Gamma|=\Order(M^2)$, 
 it is easy to see that the running time of Algorithm~\ref{algorithm:CPI} is $\Order(nM^3)$. 
 Since $M=\Order(n^2\delta^{-1})$ by line 2, we obtain the claim. 
\end{proof}

 Theorem \ref{th:CPI} is immediate from Lemma~\ref{lem:G-time} and the following Lemma~\ref{lem:G-approx}. 
\begin{lemma}\label{lem:G-approx}
\begin{eqnarray*}
 \Psi_n(1,r) \leq G_n(1,r) \leq (1 + \delta) \Psi(1,r). 
\end{eqnarray*}
\end{lemma}

 The rest of Section~\ref{sec:G-algo} proves Lemma~\ref{lem:G-approx}. 
 As a preliminary we remark the following observation from Lemma~\ref{lem:Psi_n(1,r)-lem}. 
\begin{observation}\label{obs:Psi-monotone}
 $\Psi_i(u,v)$ is monotone non-decreasing with respect to $u$, as well as $v$. 
\end{observation}
\begin{proof}
Suppose that $u\le u'$ and $v\le v'$ hold.
Lemma~\ref{lem:Psi_n(1,r)-lem} implies that 
\begin{align*}
 \Psi_i(u,v)
  &= 2^i \Pr\left[\left(\sum_{j=1}^{i} |X_{j}| \le u \right) \wedge 
           \left(\sum_{j=1}^{i}|X_{j}-c_{j}|\le v \right)\right] \\
  &\leq 2^i \Pr\left[\left(\sum_{j=1}^{i} |X_{j}| \le u' \right) \wedge 
           \left(\sum_{j=1}^{i}|X_{j}-c_{j}|\le v' \right)\right] 
  = \Psi_i(u',v').
\end{align*}
\end{proof}

First, we give a lower bound of $G_i(u,v)$. 
\begin{lemma}\label{lemma:lowG}
 $\Psi_i(u,v) \le G_i(u,v)$
for any $u,v \in \mathbb{R}$ and $i =1,2,\ldots,n$. 
\end{lemma}
\begin{proof}
 We give an inductive proof.  $\Psi_0(u,v) = G_0(u,v)$ by the definition. 
 Inductively assuming the claim for $i$, we show the claim for $i+1$ as follows:  
\begin{align*}
  G_{i+1}(u,v) 
  &= \overline{G}_{i+1}(\tfrac{1}{M}\lceil Mu \rceil, \tfrac{r}{M} \lceil \tfrac{M}{r}v \rceil) 
  && (\mbox{Recall \eqref{form:G} and \eqref{eq:remark-G}})\\
  &=\int_{-1}^1 G_i(\tfrac{1}{M}\lceil Mu \rceil -|s|, \tfrac{r}{M} \lceil \tfrac{M}{r}v \rceil -|s-c_i|){\rm d}s \\
  &\ge \int_{-1}^1 \Psi_i(\tfrac{1}{M}\lceil Mu \rceil -|s|, \tfrac{r}{M} \lceil \tfrac{M}{r}v \rceil -|s-c_i|) {\rm d}s && (\mbox{Induction hypo.})\\
  &=\Psi_{j+1}(\tfrac{1}{M}\lceil Mu \rceil, \tfrac{r}{M} \lceil \tfrac{M}{r}v \rceil)  \\
  &\geq \Psi_{j+1}(u,v)  && (\mbox{By Obs.~\ref{obs:Psi-monotone}})
\end{align*}
and we obtain the claim. 
\end{proof}

Next, we give an upper bound of $G_i(u,v)$.
\begin{lemma}
\label{lem:upG}
$G_i(u,v) \le \Psi_i(u+\frac{1}{M}i, v+\frac{r}{M}i)$ 
for any $u,v \in \mathbb{R}$ and $i =1,2,\ldots,n$. 
\end{lemma}
\begin{proof}
 The proof is an induction on $n$.
 By the definition that $G_0(u,v)=\Psi_0(u,v)$ for any $u,v$, 
   the claim is clear when $n=0$. 
 Inductively assuming the claim holds when $n=i$, 
   meaning that $G_i(u,v) \le \Psi_i(u+\frac{1}{M}j, v+\frac{r}{M}j)$ holds, 
  we show the claim when $n = i+1$. 
By the definition of $G_{i}(u,v)$ and $\overline{G}_{i+1}(u,v)$, we have
\begin{align*}
  G_{i+1}(u,v) 
  &= \overline{G}_{i+1}(\tfrac{1}{M}\lceil Mu \rceil, \tfrac{r}{M} \lceil \tfrac{M}{r}v \rceil) 
  && (\mbox{Recall \eqref{form:G} and \eqref{eq:remark-G}})\\
  &=\int_{-1}^1 G_i(\tfrac{1}{M}\lceil Mu \rceil -|s|, \tfrac{r}{M} \lceil \tfrac{M}{r}v \rceil -|s-c_i|){\rm d}s \\
  &\le \int_{-1}^1 \Psi_i(\tfrac{1}{M}\lceil Mu \rceil -|s|+\tfrac{1}{M}i, \tfrac{r}{M} \lceil \tfrac{M}{r}v \rceil -|s-c_i|+\tfrac{r}{M}i) {\rm d}s 
  && (\mbox{Induction hypo.})\\
  &\le \int_{-1}^1 \Psi_i(u+\tfrac{1}{M}-|s|+\tfrac{1}{M}i,v+\tfrac{r}{M}-|s-c_i|+\tfrac{r}{M}i) {\rm d}s  
  &&\left(\begin{array}{l}
  \mbox{By Obs.~\ref{obs:Psi-monotone}. Remark } \\ 
  \tfrac{1}{M}\lceil Mu \rceil \leq u+\tfrac{1}{M}, \\
  \tfrac{r}{M} \lceil \tfrac{M}{r}v \rceil \leq v+\tfrac{r}{M}. \end{array}\right)\\
  &= \int_{-1}^1 \Psi_i(u+\tfrac{1}{M}(i+1)-|s|,v+\tfrac{r}{M}(i+1)-|s-c_i|) {\rm d}s \\
  &=\Psi_{i+1}(u+\tfrac{1}{M}(i+1),v+\tfrac{r}{M}(i+1))  
\end{align*}
 and we obtain the claim. 
\end{proof}

\begin{lemma}\label{lem:cone}
 When $\|\Vec{c}\|_1 \leq r$, 
\begin{eqnarray*}
 \frac{\Psi(1,r)}{\Psi(1+\tfrac{n}{M},r(1+\tfrac{n}{M}))}
\geq 
 \left(\frac{M}{M+n}\right)^{2n}. 
\end{eqnarray*}
\end{lemma}
\begin{proof}
 We prove the following two inequalities, 
\begin{eqnarray}
 \frac{\Psi(1,r)}{\Psi(1+\tfrac{n}{M},r)}
&\geq &
 \left(\frac{M}{M+n}\right)^n, 
\qquad\mbox{and} \label{eq:scale1} \\
 \frac{\Psi(1+\tfrac{n}{M},r)}{\Psi(1+\tfrac{n}{M},r(1+\tfrac{n}{M}))}
&\geq &
 \left(\frac{M}{M+n}\right)^n, 
\label{eq:scale2}
\end{eqnarray}
 respectively, 
 where the proofs of \eqref{eq:scale1} and \eqref{eq:scale2} are similar. 
 The claim is clear from \eqref{eq:scale1} and \eqref{eq:scale2}. 

First we prove \eqref{eq:scale1}. 
For convenience, let
\begin{align*}
  K(q)
  = \left\{\lambda \Vec{x} \in \mathbb{R}^n 
  \mid \Vec{x} \in C(\Vec{0},1)\cap C(\Vec{c},r), \  \lambda \in \mathbb{R} \mbox{ such that } 0 \leq \lambda \leq q \right\} 
\end{align*}
 for $q \geq 1$. 
 It is not difficult to see from the definition that 
\begin{align}
 \frac{\Vol(K(1))}{\Vol(K(1+\tfrac{n}{M}))}
 = \left(\frac{M}{M+n}\right)^n
\label{eq:ratioK1}
\end{align}
 holds, 
 where we remark that $\Vol(K(1))=\Psi(1,r)$ since $K(1)=C(\Vec{0},1)\cap C(\Vec{c},r)$ by the definition.
 To claim $\Psi\!\left(1+\tfrac{n}{M},r\right)  \leq \Vol(K(1+\tfrac{n}{M}))$, 
 we show that 
\begin{eqnarray}
 C\!\left(\Vec{0},1+\tfrac{n}{M}\right) \cap C(\Vec{c},r)
 \subseteq 
  K(1+\tfrac{n}{M})
\label{claim:C1}
\end{eqnarray}
 holds. 
 Suppose $\Vec{y} \in C\!\left(\Vec{0},1+\tfrac{n}{M}\right) \cap C(\Vec{c},r)$, and 
  we prove $\Vec{y} \in K(1+\tfrac{n}{M})$. 
 More precisely,  
 let $\Vec{w} = \lambda^{-1} \Vec{y}$ where $\lambda = 1+\tfrac{n}{M}$, and 
  we show $\Vec{w} \in K(1)$. 
 Since $\Vec{y} \in C\!\left(\Vec{0},1+\tfrac{n}{M}\right)$, 
  $\|\Vec{w}\|_1 = \lambda^{-1} \|\Vec{y}\|_1 \leq 1$ holds, meaning that $\Vec{w} \in C(\Vec{0},1)$. 
 Considering $\Vec{y} \in C(\Vec{c},r)$, and the assumption $\|\Vec{c}\|_1 \leq r$, we have 
\begin{eqnarray*}
 \|\Vec{w} - \Vec{c}\|_1 
 &=& \|\lambda^{-1} \Vec{y} -\Vec{c}\|_1 \\
 &=& \|\lambda^{-1} (\Vec{y} -\Vec{c}) - (1-\lambda^{-1})\Vec{c} \|_1 \\
 &\leq& \lambda^{-1} \|\Vec{y} -\Vec{c}\|_1 + (1-\lambda^{-1}) \|\Vec{c} \|_1 \\
 &\leq& \lambda^{-1} r + (1-\lambda^{-1}) r \\
 &=& r
\end{eqnarray*}
   and hence $\Vec{w} \in C(\Vec{c},r)$. We obtain \eqref{claim:C1}. 
 Carefully recalling Lemma~\ref{lem:Psi_n(1,r)-lem},  
  $\Psi(1+\tfrac{n}{M},r) = \Vol(C\!\left(\Vec{0},1+\tfrac{n}{M}\right) \cap C(\Vec{c},r)\cap [-1,1]^n)$ holds, 
 which implies $\Psi(1+\tfrac{n}{M},r) \leq \Vol(K(1+\tfrac{n}{M}))$ with \eqref{claim:C1}. 
 Now, \eqref{eq:scale1} is easy from~\eqref{eq:ratioK1}. 

 The proof of \eqref{eq:scale2} is similar. 
Let
\begin{align*}
  K'(q)
  = \left\{\lambda (\Vec{x}-\Vec{c}) \in \mathbb{R}^n 
  \mid \Vec{x} \in C\!\left(\Vec{0},1+\tfrac{n}{M}\right) \cap C(\Vec{c},r), \  \lambda \in \mathbb{R} \mbox{ such that } 0 \leq \lambda \leq q \right\} 
\end{align*}
 for $q \geq 1$. 
 Then, 
 It is not difficult to see from the definition that 
\begin{align}
 \frac{\Vol(K'(1))}{\Vol(K'(1+\tfrac{n}{M}))}
 = \left(\frac{M}{M+n}\right)^n
\label{eq:ratioK2}
\end{align}
 holds, 
 where we remark that $\Vol(K'(1))=\Psi(1+\tfrac{n}{M},r)$. 
 To claim $\Psi\!\left(1+\tfrac{n}{M},r(1+\tfrac{n}{M})\right)  \leq \Vol(K'(1+\tfrac{n}{M}))$, 
 we show that 
\begin{eqnarray}
 C\!\left(\Vec{0},1+\tfrac{n}{M}\right) \cap C(\Vec{c},r(1+\tfrac{n}{M}))
 \subseteq 
  K'(1+\tfrac{n}{M})
\label{claim:C2}
\end{eqnarray}
 holds. 
 Suppose $\Vec{y}' \in C\!\left(\Vec{0},1+\tfrac{n}{M}\right) \cap C(\Vec{c},r(1+\tfrac{n}{M}))$, and 
  we prove $\Vec{y}' \in K'(1+\tfrac{n}{M})$. 
 More precisely,  
 let $\Vec{w}' = \lambda^{-1} (\Vec{y}'-\Vec{c})+\Vec{c}$ where $\lambda = 1+\tfrac{n}{M}$, and 
  we show $\Vec{w}' \in K'(1)$. 
 Since $\Vec{y}' \in C\!\left(\Vec{c},r(1+\tfrac{n}{M})\right)$, 
  $\|\Vec{w}'-\Vec{x}\|_1 = \lambda^{-1} \|\Vec{y}'-\Vec{c}\|_1 \leq 1$ holds, meaning that $\Vec{w}' \in C(\Vec{0},1)$. 
 Considering $\Vec{y}' \in C(\Vec{0},1+\tfrac{n}{M})$, and the assumption $\|\Vec{c}\|_1 \leq r$, we have 
\begin{eqnarray*}
 \|\Vec{w}'\|_1 
 &=& \|\lambda^{-1} (\Vec{y}' -\Vec{c}) + \Vec{c}\|_1 \\
 &=& \|\lambda^{-1} \Vec{y}' + (1-\lambda^{-1})\Vec{c} \|_1 \\
 &\leq& \lambda^{-1} \|\Vec{y}\|_1 + (1-\lambda^{-1}) \|\Vec{c} \|_1 \\
 &\leq& \lambda^{-1} \left(1+\tfrac{n}{M} \right) + (1-\lambda^{-1}) r \\
 &\leq& \lambda^{-1} \left(1+\tfrac{n}{M} \right) + (1-\lambda^{-1}) \left(1+\tfrac{n}{M} \right) 
 \hspace{3em} (\mbox{since $r \leq 1$}) \\
 &=& \left(1+\tfrac{n}{M} \right)
\end{eqnarray*}
   and hence $\Vec{w}' \in C\!\left(\Vec{0},1+\tfrac{n}{M} \right)$. 
 We obtain \eqref{claim:C2}, and hence \eqref{eq:scale2} from \eqref{eq:ratioK2}. 
 Now, we obtain the claim. 
\end{proof}

 Now, we prove Lemma~\ref{lem:G-approx}. 
\begin{proof}[Proof of Lemma~\ref{lem:G-approx}]
 The first inequality is immediate from Lemma~\ref{lemma:lowG}. 
 Then, we show the latter inequality.  
Lemma~\ref{lem:cone} implies that 
\begin{eqnarray*}
\frac{\Psi_n(1,r)}{\Psi_n(1+\tfrac{n}{M},r(1+\tfrac{n}{M}))}
 &\geq& \left(\frac{M}{M+n} \right)^{2n} 
 = \left(\frac{1}{1+\frac{n}{M}} \right)^{2n} \\
 &\geq& \left(1-\frac{n}{M} \right)^{2n} 
  \hspace{4em} \left(\mbox{since $\left(1+\frac{n}{M} \right)^{2n}\left(1-\frac{n}{M} \right)^{2n} \leq 1$}\right) \\
 &\geq& \left(1-\frac{\delta}{4n} \right)^{2n} 
  \hspace{4em} \left(\mbox{since $M \geq 4n^2\delta^{-1}$}\right) \\
 &\geq& 1-2n\frac{\delta}{4n} 
  = 1-\frac{\delta}{2}
\end{eqnarray*}
 holds. Thus, 
\begin{eqnarray*}
\frac{\Psi_n(1+\tfrac{n}{M},r(1+\tfrac{n}{M}))}{\Psi_n(1,r)}
\leq \frac{1}{1-\frac{\delta}{2}}
\leq 1+\delta
\end{eqnarray*}
 for any $\delta \leq 1$, and we obtain the claim. 
\end{proof}

\section{Hardness of the Volume of the Intersection of Two Cross-polytopes}\label{sec:hardness}
 This section establishes the following. 
\begin{theorem}\label{thm:hardness}
 Given a vector $\Vec{c} \in\mathbb{Z}_{>0}^n$ and integers $r_1,r_2 \in \mathbb{Z}_{>0}$, 
  computing the volume of $C(\Vec{0},r_1) \cap C(\Vec{c},r_2)$ is {\#}P-hard, 
  even when each cross-polytopes contains the center of the other one, 
 i.e., 
  $\Vec{0} \in C(\Vec{c},r_2)$ and  
  $\Vec{c} \in C(\Vec{0},r_1)$.  
\end{theorem}
 The proof of Theorem~\ref{thm:hardness} is a reduction of counting set partitions, which is a well-known {\#}P-hard problem. 
\subsection{Idea for the reduction}\label{sec:hard1}
 To be precise, 
  we reduce the following problem, which is a version of counting set partition. 
\begin{problem}[{\#}LARGE SET]\label{prob:LSet}
 Given an integer vector $\Vec{a} \in \mathbb{Z}_{>0}^n$ 
  such that $\|\Vec{a}\|_1$ is even, meaning that $\|\Vec{a}\|_1/2$ is an integer, 
 the problem is to compute 
\begin{eqnarray}
 \left| \{ \Vec{\sigma} \in \{-1,1\}^n \mid \langle \Vec{\sigma}, \Vec{a} \rangle > 0 \} \right|. 
\label{def:largeset1}
\end{eqnarray}
\end{problem}
 Note that 
\begin{eqnarray*}\textstyle
  \left| \{ \Vec{\sigma} \in \{-1,1\}^n \mid \langle \Vec{\sigma}, \Vec{a} \rangle = 0 \} \right|   
=
  \left| \left\{ S \subseteq \{1,\ldots,n\} \mid  \sum_{i \in S} a_i = \frac{\|\Vec{a}\|_1}{2} \right\} \right|   
\end{eqnarray*}
  holds:
  if $\Vec{\sigma} \in \{-1,1\}^n$ satisfies $\langle \Vec{\sigma}, \Vec{a} \rangle = 0$, 
  then let $S \subseteq \{1,\ldots,n\}$ be the set of indices of $\sigma_i=1$ 
 then 
  $\sum_{i \in S} a_i = \|\Vec{a}\|_1/2$ holds. 
 Using the following simple observation, 
  we see that Problem~\ref{prob:LSet} is equivalent to counting set partitions. 
\begin{observation}\label{obs:set-parition}
 For any $\Vec{\sigma} \in \{-1,1\}^n$,  
  $\langle \Vec{\sigma},\Vec{a} \rangle > 0$ if and only if $\langle -\Vec{\sigma}, \Vec{a} \rangle < 0$. 
\end{observation}
 By Observation~\ref{obs:set-parition}, we see that 
\begin{eqnarray}
 \left| \{ \Vec{\sigma} \in \{-1,1\}^n \mid \langle \Vec{\sigma}, \Vec{a} \rangle = 0 \} \right| 
 = 2^n -2\left| \{ \Vec{\sigma} \in \{-1,1\}^n \mid \langle \Vec{\sigma}, \Vec{a} \rangle > 0 \} \right|. 
\end{eqnarray}

\begin{figure}
\begin{center}
\includegraphics[bb= 10 20 610 420, scale=0.5, clip]{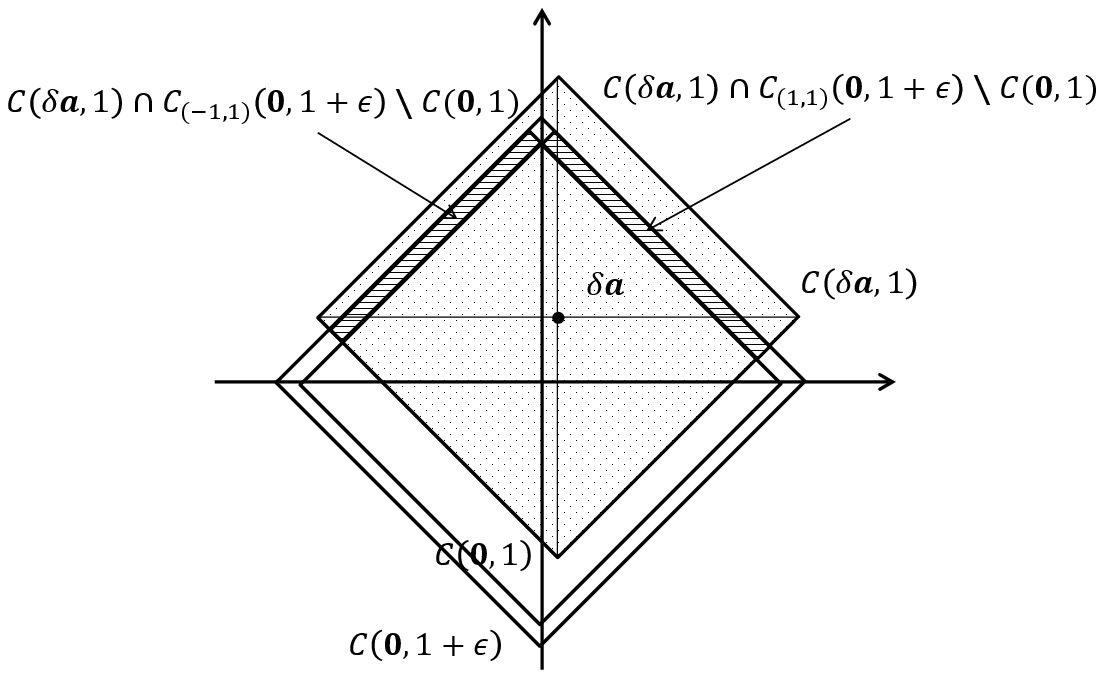}
\end{center}
\caption{$C(\Vec{0},1+\epsilon) \cap C(\delta\Vec{a},1) \setminus C(\Vec{0},1)$. }\label{fig:intersection}
\end{figure}

 In the following, let 
  $\Vec{a} \in \mathbb{Z}_{> 0}^n$ be an instance of Problem~\ref{prob:LSet}. 
 Roughly speaking, 
  our proof of Theorem~\ref{thm:hardness} claims that 
\begin{eqnarray}
\Vol \left(  C(\delta \Vec{a},1) \cap C(\Vec{0},1+\epsilon) \setminus  C(\Vec{0},1) \right)
 \sim |\{ \Vec{\sigma} \in \{-1,1\} \mid \langle \Vec{\sigma},\Vec{a} \rangle > 0 \}|
\label{eq:hard-intuition}
\end{eqnarray}
 holds (see Figure~\ref{fig:intersection}), when $0 < \epsilon < \delta \ll 1 / \|\Vec{a}\|_1$. 
 For convenience, we define 
\begin{eqnarray}
 C_{\Vec{\sigma}}(\Vec{c},r) = \{ \Vec{x} \in C(\Vec{c},r) \mid \sigma_i (x_i-c_i) \geq  0\, (i = 1,\ldots,n)\}
 \label{def:sigma-cone}
\end{eqnarray}
 for any $\Vec{\sigma} \in \{-1,1\}^n$. 
 Note that 
\begin{eqnarray}
 C(\delta \Vec{a},1) \cap C(\Vec{0},1+\epsilon) \setminus  C(\Vec{0},1)
  =  
  \dot\bigcup_{\Vec{\sigma} \in \{-1,1\}^n}
 C(\delta \Vec{a},1) \cap C_{\Vec{\sigma}}(\Vec{0},1+\epsilon) \setminus  C(\Vec{0},1)
\end{eqnarray}
 holds. 
 In the following, we claim  for each $\Vec{\sigma} \in \{-1,1\}$ that 
\begin{eqnarray*}
 \Vol \left( C(\delta \Vec{a},1) \cap C_{\Vec{\sigma}}(\Vec{0},1+\epsilon) \setminus  C(\Vec{0},1) \right)
 \simeq \begin{cases}
   0 & (\mbox{if} \langle \Vec{\sigma},\Vec{a} \rangle \leq 0) \\
  \dfrac{\epsilon}{(n-1)!} & (\mbox{otherwise})
 \end{cases}
\end{eqnarray*}
 with appropriate $\epsilon$ and $\delta$. 

 First, we consider the case that $\Vec{\sigma} \in \{-1,1\}$ satisfies $\langle \Vec{\sigma},\Vec{a} \rangle \leq 0$. 
We define 
\begin{eqnarray}
 H^-_{\Vec{\sigma}} (\Vec{c},r) 
  &\defeq& \{\Vec{x} \in \mathbb{R}^n \mid \langle \Vec{\sigma},\Vec{x}-\Vec{c} \rangle \leq r \} 
\label{def:H-} \\
 H^+_{\Vec{\sigma}} (\Vec{c},r) 
  &\defeq& \{\Vec{x} \in \mathbb{R}^n \mid \langle \Vec{\sigma},\Vec{x}-\Vec{c} \rangle > r \} 
\label{def:H+} \\
 H_{\Vec{\sigma}} (\Vec{c},r) 
  &\defeq& \{\Vec{x} \in \mathbb{R}^n \mid \langle \Vec{\sigma},\Vec{x}-\Vec{c} \rangle = r \} 
\label{def:H}
\end{eqnarray}
 for convenience (see Figure~\ref{fig:H}). 
\begin{figure}
\begin{center}
 \includegraphics[bb= 10 20 610 480, scale=0.5, clip]{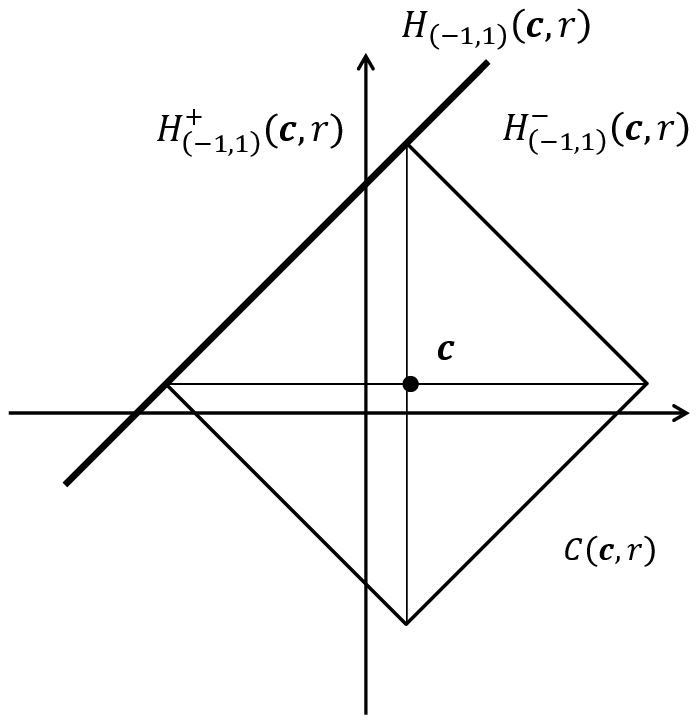}
\end{center}
\caption{$H^-_{\Vec{\sigma}} (\Vec{c},r) $, $H^+_{\Vec{\sigma}} (\Vec{c},r) $, $H_{\Vec{\sigma}} (\Vec{c},r) $. }\label{fig:H}
\end{figure}

\subsection{Facet for $\langle \Vec{\sigma},\Vec{a} \rangle \leq 0$}
\begin{proposition}\label{prop:hard-case1}
 If $\langle \Vec{\sigma},\Vec{a} \rangle \leq 0$, 
 then $C(\delta \Vec{a},1)$ is in the half-space  $H^-_{\Vec{\sigma}} (\Vec{0},1)$. 
\end{proposition}
\begin{proof}
 Notice that $\Vec{x} \in C(\delta \Vec{a},1)$ implies 
\begin{eqnarray}
 \langle \Vec{x}-\delta\Vec{a},\Vec{\sigma} \rangle \leq 1
\label{eq:20160404pp1}
\end{eqnarray}
 holds. 
 Since the hypothesis that $\langle \Vec{a},\Vec{\sigma} \rangle=0$, 
\begin{eqnarray*}
 \langle \Vec{\sigma}, \Vec{x}-\delta\Vec{a} \rangle 
 =
 \langle \Vec{\sigma}, \Vec{x} \rangle 
 -\delta \langle \Vec{\sigma}, \Vec{a} \rangle 
 =
 \langle \Vec{\sigma},\Vec{x}\rangle 
\end{eqnarray*}
 holds, which implies with \eqref{eq:20160404pp1} that 
\begin{eqnarray}
 \langle \Vec{\sigma},\Vec{x} \rangle \leq 1. 
\end{eqnarray}
 We obtain the claim. 
\end{proof}

 Proposition~\ref{prop:hard-case1} implies 
  $ C_{\Vec{\sigma}}(\Vec{0},1+\epsilon) \setminus C(\Vec{0},1) \subset H^+_{\Vec{\sigma}} (\Vec{0},1)$, 
  and we see the following (see also Figure~\ref{fig:sigma-a-0}). 
\begin{corollary}\label{cor:hard-case1}	
 If $\langle \Vec{\sigma},\Vec{a} \rangle \leq 0$, 
 then $\Vol \left( C(\delta \Vec{a},1) \cap C_{\Vec{\sigma}}(\Vec{0},1+\epsilon) \setminus  C(\Vec{0},1) \right) =0$. 
\end{corollary}
\begin{figure}
\begin{center}
  \includegraphics[bb= -50 469 610 824, scale=0.5, clip]{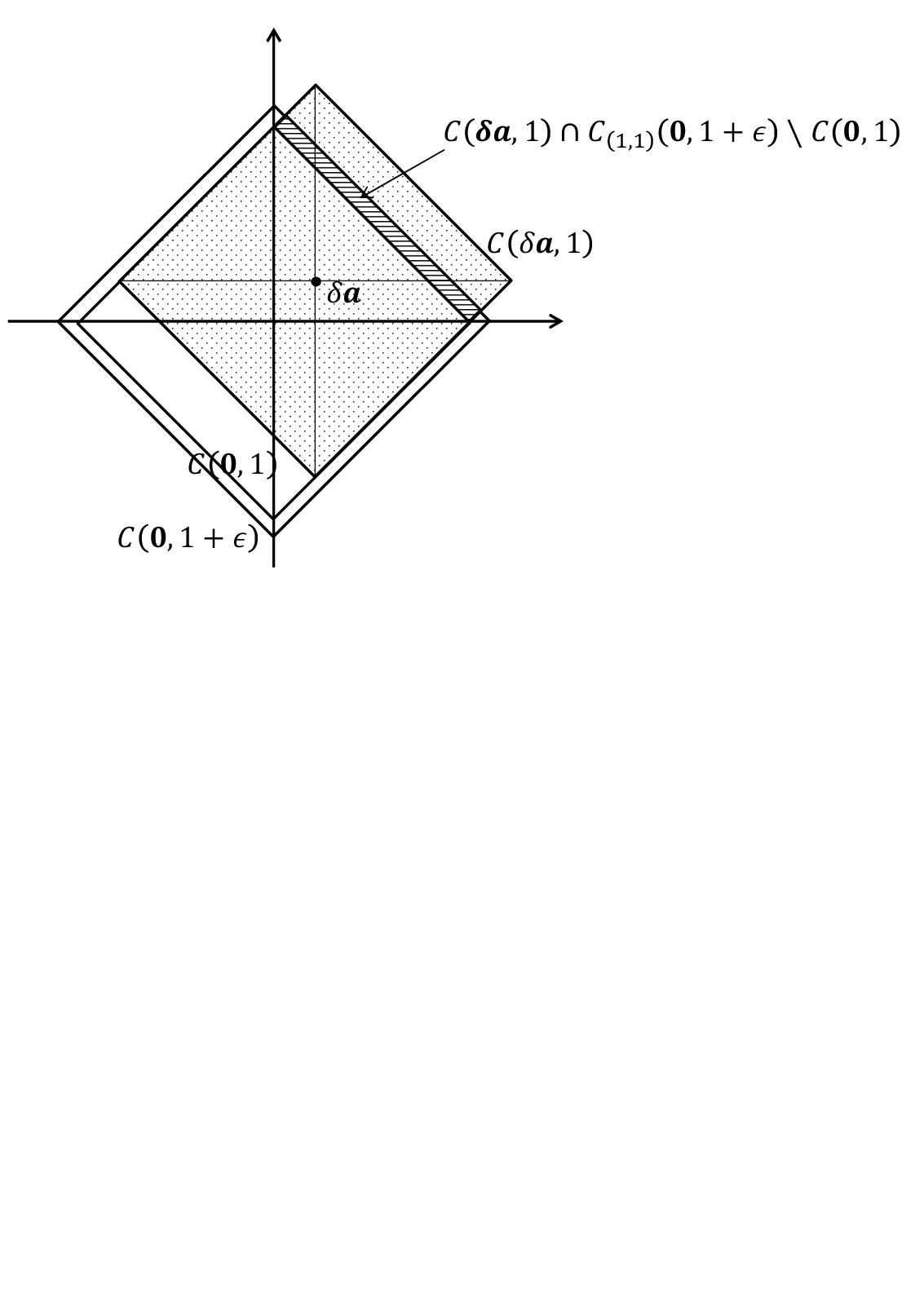}
\end{center}
\caption{$C(\Vec{0},1+\epsilon) \cap C(\delta\Vec{a},1) \setminus C(\Vec{0},1)$ 
  where $\langle \Vec{\sigma},\Vec{a} \rangle=0$ holds for $\Vec{\sigma}=(1,-1)$ ($(-1,1)$ as well).}\label{fig:sigma-a-0}
\end{figure}

\subsection{Facet for $\langle \Vec{\sigma},\Vec{a} \rangle > 0$}
Next, we are concerned with the case that $\Vec{\sigma}\in \{-1,1\}$ satisfies $\langle \Vec{\sigma},\Vec{a}\rangle > 0$. 
 Notice that 
  $ H_{\Vec{\sigma}} (\Vec{0},1+ \epsilon)$ and $ H_{\Vec{\sigma}} (\delta \Vec{a},1)$ are in parallel 
 since they have a common normal vector $\Vec{\sigma}$. 
\begin{proposition}\label{prop:komikado}
 Suppose that $\langle \Vec{\sigma},\Vec{a} \rangle > 0$ holds. 
 If $\epsilon < \delta $ then 
  $ H^-_{\Vec{\sigma}} (\Vec{0},1+ \epsilon) \subseteq H^-_{\Vec{\sigma}} (\delta \Vec{a},1)$. 
\end{proposition}
\begin{proof}
 Let $\Vec{x} \in H^-_{\Vec{\sigma}} (\Vec{0},1+ \epsilon)$, then 
\begin{eqnarray}
  \langle \Vec{\sigma},\Vec{x}\rangle \leq 1+ \epsilon. 
\label{20160404pp2}
\end{eqnarray}
 holds. 
 Suppose for a contradiction that  $\Vec{x} \not\in H^-_{\Vec{\sigma}} (\delta \Vec{a},1))$, 
 then 
\begin{eqnarray*}
  \langle \Vec{\sigma},\Vec{x}- \delta \Vec{a} \rangle 
 =\langle \Vec{\sigma},\Vec{x} \rangle 
   - \delta \langle \Vec{\sigma},\Vec{a} \rangle 
 > 1
\end{eqnarray*}
  holds. 
 It implies 
\begin{eqnarray}
 \langle \Vec{\sigma},\Vec{x} \rangle 
 > 1+ \delta \langle \Vec{\sigma},\Vec{a} \rangle  
 \geq 1+ \delta 
\label{20160404pp3}
\end{eqnarray}
 holds since $\langle \Vec{\sigma},\Vec{a} \rangle > 0$ means $\langle \Vec{\sigma},\Vec{a} \rangle \geq 1$. 
 Clearly, \eqref{20160404pp3} and \eqref{20160404pp2}
  contradict to the hypothesis that $\epsilon < \delta$. 
\end{proof}
 Proposition~\ref{prop:komikado} implies that 
  $C(\delta\Vec{a},1) \cap C_{\Vec{\sigma}}(\Vec{0},1+\epsilon) \setminus C(\Vec{0},1) \neq \emptyset$. 
 More precisely, 
  we observe the following, which we will use later.  
\begin{observation}\label{obs:takasa}
 The $L_2$ distance between 
  $ H_{\Vec{\sigma}} (\Vec{0},1+ \epsilon)$ and 
  $ H_{\Vec{\sigma}} (\Vec{0},1)$ is $\dfrac{\epsilon}{\sqrt{n}}$. 
\end{observation}
 The volume of $C(\delta\Vec{a},1) \cap C_{\Vec{\sigma}}(\Vec{0},1+\epsilon) \setminus C(\Vec{0},1)$
  is evaluated as follows, where we assume that $\epsilon$ is sufficiently small. 
\begin{proposition}\label{prop:hard-upper}
 Suppose that $\langle \Vec{\sigma},\Vec{a} \rangle >0$ holds. 
 If $\epsilon < \delta$ then 
\begin{eqnarray*}
 \Vol\left( C(\delta\Vec{a},1) \cap C_{\Vec{\sigma}}(\Vec{0},1+\epsilon) \setminus C(\Vec{0},1) \right)
  \leq  \frac{\epsilon}{(n-1)!} (1+\epsilon)^{n-1}. 
\end{eqnarray*} 
\end{proposition}
\begin{proof}
\begin{eqnarray*}
 \Vol\left( C(\delta\Vec{a},1) \cap C_{\Vec{\sigma}}(\Vec{0},1+\epsilon) \setminus C(\Vec{0},1) \right)
 &\leq& \frac{1}{n!} \left((1+\epsilon)^n -1 \right) \\ 
 &=& \frac{ \epsilon }{n!} \sum_{i=0}^{n-1} (1+\epsilon)^i\\ 
 &\leq& \frac{\epsilon}{n!} n(1+\epsilon)^{n-1} \\ 
 &=& \frac{ \epsilon }{(n-1)!}(1+\epsilon)^{n-1}. 
\end{eqnarray*} 
\end{proof}

 Next, we give a lower bound of 
  $\Vol\left( C(\delta\Vec{a},1) \cap C_{\Vec{\sigma}}(\Vec{0},1+\epsilon) \setminus C(\Vec{0},1) \right)$. 
 To begin with, we observe the following. 
\begin{observation}\label{obs:2}
 For any vertex $\Vec{v} \in \{\pm \Vec{e}_1, \ldots, \pm \Vec{e}_n \}$ of $C(\Vec{0},1)$, 
  the nearest vertex of $C(\delta\Vec{a},1)$ is 
  in the $L_2$ distance $\delta \|\Vec{a}\|_2$.  
\end{observation}
 Observation~\ref{obs:2} implies the following. 
\begin{proposition}\label{prop:obs2}
 For any vertex $\Vec{v} \in \{\pm \Vec{e}_1, \ldots, \pm \Vec{e}_n \}$ of $C(\Vec{0},1)$, 
  one of hyperplanes $ H_{\Vec{\sigma}} (\delta \Vec{a},1)$ and 
   $ H_{-\Vec{\sigma}} (\delta \Vec{a},1)$ 
   is in the $L_2$ distance $\delta \|\Vec{a}\|_2$ for any $\Vec{\sigma} \in \{-1,1\}^n$. 
\qed
\end{proposition}
 Proposition~\ref{prop:obs2} implies that 
  when $\langle \Vec{\sigma},\Vec{a} \rangle > 0$, 
  i.e., $C(\delta\Vec{a},1) \cap C_{\Vec{\sigma}}(\Vec{0},1+\epsilon) \setminus C(\Vec{0},1) \neq \emptyset$ holds, 
  $H_{\Vec{\sigma}'}(\delta \Vec{a},1)$ shaves off only a few area of  
  $ H_{\Vec{\sigma}} (\Vec{0},1) \cap C(\delta\Vec{a},1) \cap C_{\Vec{\sigma}}(\Vec{0},1+\epsilon)$. 
\begin{figure}
\begin{center}
\includegraphics[bb= -10 10 610 430, scale=0.5, clip]{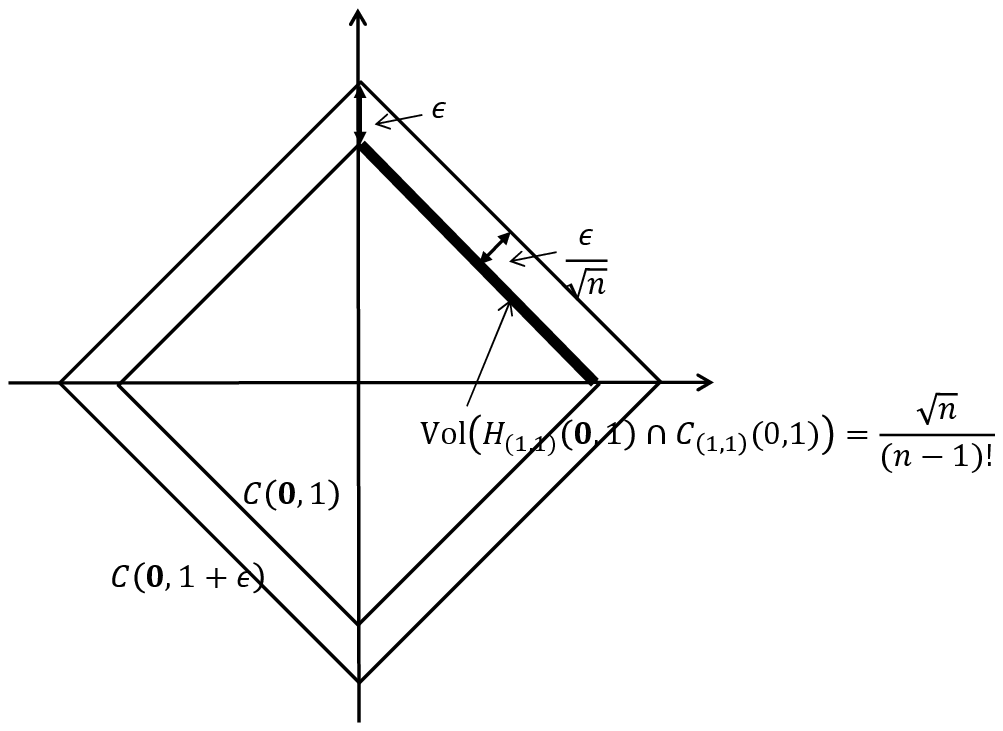}
\end{center}
\caption{For Observation~\ref{obs:takasa} and Proposition~\ref{prop:teihen}.}\label{fig:vol'}
\end{figure}
 It is formally described as follows. 
\begin{proposition}\label{prop:teihen}
 When $\langle \Vec{\sigma},\Vec{a} \rangle > 0$, 
\begin{eqnarray*}
 \Vol' (H_{\Vec{\sigma}} (\Vec{0},1) \cap C_{\Vec{\sigma}}(\Vec{0},1) \cap C(\delta\Vec{a},1))
  \geq  \frac{\sqrt{n}}{(n-1)!} (1-\delta\|\Vec{a}\|_1)^{n-1}
\end{eqnarray*}
 holds, 
 where $\Vol'(S)$ denotes the $n-1$ dimensional volume of $S \subseteq \mathbb{R}^{n-1}$. 
\end{proposition}
\begin{proof}
 Remark that 
\begin{eqnarray*}
 \Vol' (H_{\Vec{\sigma}} (\Vec{0},1) \cap C_{\Vec{\sigma}}(\Vec{0},1))
  =  \frac{\sqrt{n}}{(n-1)!} 
\end{eqnarray*}
 since 
$ \frac{1}{n}\frac{1}{\sqrt{n}}\Vol' (H_{\Vec{\sigma}} (\Vec{0},1) \cap C_{\Vec{\sigma}}(\Vec{0},1))
  = \Vol (C_{\Vec{\sigma}}(\Vec{0},1))
  =  \frac{1}{n!} $. Thus, 
\begin{eqnarray*}
 \Vol' (H_{\Vec{\sigma}} (\Vec{0},1) \cap C_{\Vec{\sigma}}(\Vec{0},1) \cap C(\delta\Vec{a},1))
 &\geq& \frac{\sqrt{n}}{(n-1)!}  (1-\delta\|\Vec{a}\|_2)^{n-1} \\
 &\geq& \frac{\sqrt{n}}{(n-1)!}  (1-\delta\|\Vec{a}\|_1)^{n-1}
\end{eqnarray*}
 where the last inequality follows the fact $\|\Vec{a}\|_2 \leq \|\Vec{a}\|_1$. 
\end{proof}

 Observation~\ref{obs:takasa} and Proposition~\ref{prop:teihen} implies the following. 
\begin{corollary}\label{kenske}
 Suppose that $\langle \Vec{\sigma},\Vec{a} \rangle > 0$ holds. 
 If $\epsilon < \delta$ then 
\begin{eqnarray*}
 \Vol\left( C(\delta\Vec{a},1) \cap C_{\Vec{\sigma}}(\Vec{0},1+\epsilon) \setminus C(\Vec{0},1) \right)
 \geq \frac{\epsilon}{(n-1)!} (1-\delta\|\Vec{a}\|_1)^{n-1}. 
\end{eqnarray*} 
\end{corollary}

\subsection{Proof of Theorem~\ref{thm:hardness}}
\begin{proposition}\label{prop:hard-lower}
 Suppose that $\langle \Vec{\sigma},\Vec{a} \rangle >0$ holds. 
 If $\epsilon < \delta$ and $\delta < \dfrac{0.1}{n 2^n \|\Vec{a}\|_1}$, 
  then 
\begin{eqnarray*}
 \Vol\left( C(\delta\Vec{a},1) \cap C_{\Vec{\sigma}}(\Vec{0},1+\epsilon) \setminus C(\Vec{0},1) \right)
  \geq 
 \dfrac{\epsilon }{(n-1)!} \left(1-\dfrac{0.1}{2^n}\right). 
\end{eqnarray*}
\end{proposition}
\begin{proof}
By Corollary~\ref{kenske}, 
\begin{eqnarray*}
 \Vol\left( C(\delta\Vec{a},1) \cap C_{\Vec{\sigma}}(\Vec{0},1+\epsilon) \setminus C(\Vec{0},1) \right)
 & \geq &
 \dfrac{\epsilon}{(n-1)!} \left(1-\delta \|\Vec{a}\|_1 \right)^{n-1} \\
 & \geq &
 \dfrac{\epsilon}{(n-1)!} \left(1- (n-1)\delta \|\Vec{a}\|_1 \right) \\
 & \geq &
 \dfrac{\epsilon}{(n-1)!} \left(1-\dfrac{0.1}{2^n}\right). 
\end{eqnarray*}
\end{proof}

 Now, we revisit the upper bound. 
 When $\epsilon$ is small enough, 
 Proposition~\ref{prop:hard-upper} implies the following. 
\begin{proposition}\label{prop:hard-upper2}
 Suppose that $\langle \Vec{\sigma},\Vec{a} \rangle >0$ holds. 
 If $\epsilon < \delta$ and $\delta < \dfrac{0.1}{n 2^n \|\Vec{a}\|_1}$, 
  then 
\begin{eqnarray*}
 \Vol\left( C(\delta\Vec{a},1) \cap C_{\Vec{\sigma}}(\Vec{0},1+\epsilon) \setminus C(\Vec{0},1) \right)
  \leq  \frac{\epsilon}{(n-1)!} \left(1+\dfrac{0.1}{2^n}\right). 
\end{eqnarray*} 
\end{proposition}
\begin{proof}
 Recall Proposition~\ref{prop:hard-upper}, which implies 
\begin{eqnarray*}
 \Vol\left( C(\delta\Vec{a},1) \cap C_{\Vec{\sigma}}(\Vec{0},1+\epsilon) \setminus C(\Vec{0},1) \right)
  \leq  \frac{\epsilon}{(n-1)!}  (1+\epsilon)^{n-1}
\end{eqnarray*} 
 under the hypothesis. 
 Remark that $\|\Vec{a}\|_1 \geq n$ since $\Vec{a} \in \mathbb{Z}_{>0}^n$. 
 Thus $\epsilon < \dfrac{0.1}{n^2 2^n}$ is assumed 
  by the hypothesis, and hence 
\begin{eqnarray*}
 \frac{\epsilon (1+\epsilon)^{n-1}}{(n-1)!} 
 & \leq &
 \dfrac{\epsilon}{(n-1)!} \left(1 + \dfrac{0.1}{n^2 2^n} \right)^n \\
 & \leq &
 \dfrac{\epsilon}{(n-1)!} \left(1 +  \dfrac{0.1}{2^n}\right). 
\end{eqnarray*}
\end{proof}

Corollary~\ref{cor:hard-case1}, 
Propositions~\ref{prop:hard-lower} and~\ref{prop:hard-upper2} imply the following. 
\begin{lemma}\label{lem:hard-core}
 Suppose that $\epsilon < \delta$ and $\delta < \dfrac{0.1}{n 2^n \|\Vec{a}\|_1}$ hold. 
 Let 
\begin{eqnarray*}
 Z := 
  \dfrac{\Vol \left(  C(\delta \Vec{a},1) \cap C(\Vec{0},1+\epsilon) \setminus  C(\Vec{0},1) \right)}{\frac{\epsilon}{(n-1)!}}
\end{eqnarray*}
  then 
\begin{eqnarray}
 Z-0.1
\leq
 |\{ \Vec{\sigma} \in \{-1,1\} \mid \langle \Vec{\sigma},\Vec{a} \rangle > 0 \}|
\leq
Z+0.1
\label{eq:hard-z}
\end{eqnarray}
\end{lemma}
\begin{proof}
 Proposition~\ref{prop:hard-lower} implies that 
\begin{eqnarray*}
 \frac{\Vol\left( C(\delta\Vec{a},1) \cap C_{\Vec{\sigma}}(\Vec{0},1+\epsilon) \setminus C(\Vec{0},1) \right)}{\frac{\epsilon}{(n-1)!}}
  \geq 
 1-\dfrac{0.1}{2^n} 
\end{eqnarray*}
 holds. Clearly, 
  $|\{ \Vec{\sigma} \in \{-1,1\} \mid \langle \Vec{\sigma},\Vec{a} \rangle > 0 \}| \leq 2^n$, 
 we obtain the lower bound of~\eqref{eq:hard-z}. 
 The upper bound is similar. 
\end{proof}
\begin{corollary}\label{cor:hard-premain}
 Suppose that $\epsilon < \delta$ and $\delta < \dfrac{0.1}{n 2^n \|\Vec{a}\|_1}$ hold. 
 Then, 
\begin{eqnarray*}
\left[ \frac{(n-1)!}{\epsilon} \Vol \left(  C(\delta \Vec{a},1) \cap C(\Vec{0},1+\epsilon) \setminus  C(\Vec{0},1) \right) 
 \right] = 
 |\{ \Vec{\sigma} \in \{-1,1\} \mid \langle \Vec{\sigma},\Vec{a} \rangle > 0 \}| 
\end{eqnarray*}
 where $[x]$ for $x \in \mathbb{R}$ denotes the integer $z$ minimizing $|z-x|$. 
\end{corollary}

 To make values integer, set $\delta = 1/r$ and $\epsilon = \delta/2$, then we obtain the following. 
\begin{lemma}\label{lem:hard-main1}
 Let $r=10n 2^n \|\Vec{a}\|_1$, then 
\begin{eqnarray*}
\left[ (n-1)!  \Vol \left(  C(2\Vec{a},2r) \cap C(\Vec{0},2r+1) \setminus  C(\Vec{0},2r) \right) 
 \right] = 
 |\{ \Vec{\sigma} \in \{-1,1\} \mid \langle \Vec{\sigma},\Vec{a} \rangle > 0 \}|. 
\end{eqnarray*}
\end{lemma}
 Finally, we remark that 
  $\lg(r) = \Order(n \log n + \log \|\Vec{a}\|_1)$, 
  meaning that the reduction is in time polynomial in $n$ and $\log \|\Vec{a}\|_1$, which is the input size of Problem~\ref{prob:LSet}. 
 Now we obtain Theorem~\ref{thm:hardness}.

\section{Intersection of a Constant Number of Cross-polytopes}\label{sec:extension}
This section extends the algorithm in Section \ref{sec:fptas-two-cp} 
to the intersection of 
$k$ cross-polytopes for any {\em constant} $k\in \Integer_+$. 
Let $\Vec{p}_i\in \Real^n$, $r_i\in \Real_{\ge 0}$ and 
$C(\Vec{p}_i,r_i)$ for $i=1,\dots,k$, 
where $C(\Vec{p},r)$ is a cross-polytope ($L_1$-ball) 
with center $\Vec{p}\in \Real^n$ and radius $r\in \Real_{\ge 0}$.
Then, we are to compute the following polytope given by
\begin{align}
  S(\Pi,\Vec{r})= \bigcap_{i=1}^k C(\Vec{p}_i,r_i),
\end{align}
where $\Pi$ is 
an $n\times k$ matrix $\Pi=(\Vec{p}_1,\dots,\Vec{p}_k)$ and
$\Vec{r}=(r_1,\dots,r_k)$.
For the analysis,
we assume that $\Vec{p}_1,\dots,\Vec{p}_k$ are internal points of 
$S(\Pi,\Vec{r})$.
Without loss of generality, we assume that $\Vec{p}_1=\Vec{0}$ and 
$\Vec{0}\le \Vec{r} \le \Vec{1}$.

We prove the following theorem.
\begin{theorem}
  \label{th:CPI_k}
  There is an algorithm that outputs an approximation $Z$ of ${\rm Vol}(S(\Pi,\Vec{r}))$ in 
  $O(k^{k+2}n^{2k+3}/\delta^{k+1})$ time satisfying ${\rm Vol}(S(\Pi,\Vec{r}))\le Z \le (1+\delta){\rm Vol}(S(\Pi,\Vec{r}))$.
\end{theorem}

\subsection{Algorithm description}
We explain the idea of our algorithm for approximating 
${\rm Vol}(S(\Pi,\Vec{r}))$ as follows.
First, ${\rm Vol}(S(\Pi,\Vec{r}))$ is given by the following probagility
\begin{align}
  {\rm Vol}(S(\Pi,\Vec{r})= 2^n \Pr\left[\bigwedge_{i=1}^k \|\Vec{X}-\Vec{p}_i\|_1 \le u_i\right], \label{form:probability_k}
\end{align}
where $\Vec{X}=(X_1,\dots,X_n)$ is a uniform random vector over $[-1,1]^n$.
We rewrite the probability as the repetition of an integral formula.
Then, we staircase approximate the integral.

To transform (\ref{form:probability_k}) into the repetition of 
an integral formula, for $\Pi$ and $\Vec{u}\in \mathbb{R}^k$, we define
\begin{align}
\Psi_j(\Pi,\Vec{u})=2^j\Pr\left[\bigwedge_{i=1}^k \|\Vec{X}-\Vec{p}_i\|_1 \le u_i\right],
\end{align}
so that we have ${\rm Vol}(S(\Pi,\Vec{r}))=\Psi_n(\Pi,\Vec{r})$.
We have $\Psi_0(\Pi,\Vec{u})=1$ if $\Vec{u}\ge \Vec{0}$ and $\Psi_0(\Pi, \Vec{u})=0$ otherwise.
We can obtain $\Psi_j(\Pi,\Vec{u})$ from $\Psi_{j-1}(\Pi,\Vec{u})$ by
\begin{align*}
  \Psi_j(\Pi,\Vec{u})&=\int_{x_i\in [-1,1]} \Psi_{j-1}(\Pi,\Vec{u}-\Vec{q}_j(x_j)) {\rm d}x_j,
\end{align*}
where $\Vec{q}_j(x_j)=(|x_j-p_{1,j}|,\dots,|x_j-p_{k,j}|)$.
Although this gives a simple expression for $S(\Pi,\Vec{r})$, it is hard 
to compute the repetition of the integral because there are exponentially
many breakpoints of the derivative of $\Psi_n(\Pi,\Vec{u})$ of some order.

We compute the staircase approximation $G_j(\Pi,\Vec{u})$ of 
$\Psi_j(\Pi,\Vec{u})$ as follows. 
For convenience, 
we consider an intermediate $\overline{G}_j(\Pi,\Vec{u})$ given by
\begin{align}
  \overline{G}_j(\Pi,\Vec{u}) = \int_{s\in [-1,1]} G_{j-1}(\Pi,\Vec{u}-\Vec{q}_j(s)){\rm d}s. \label{form:Gbar_k}
\end{align}
This integral can be reduced to a sum, which we will explain 
after we define $G_j(\Pi,\Vec{u})$ for $j=1,\dots,n$. 
After that, $G_j(\Pi,\Vec{u})$ is a staircase approximation of $\overline{G}_j(\Pi,\Vec{u})$
given by
\begin{align}
  G_j(\Pi,\Vec{u})= \overline{G}_j \left(\Pi, \lceil M\Vec{u}/\Vec{r} \rceil/M\right)\label{form:G_k}
\end{align}
where $\lceil M\Vec{u}/\Vec{r}\rceil$ means a vector 
$(\lceil Mu_1/r_1\rceil, \dots, \lceil Mu_k/r_k\rceil)$,
and $M=2kn^2/\delta$ is a parameter of our Algorithm \ref{algorithm:CPI_k} that
is shown later.
Note that the computation of $G_j(\Pi,\Vec{u})$ is actually 
the computation of $(M+1)^k$ values.
Since $\Vec{u}-\Vec{q}_j(s)\le \Vec{r}$ holds in the computation of 
(\ref{form:Gbar2_k}) as long as $\Vec{u}\le \Vec{r}$,
we need not to have the value for the cases where
$\Vec{u}\le \Vec{r}$ does not hold.

Let us see that the integral for computing $\overline{G}_j(\Pi,\Vec{u})$ 
can be transformed into a sum. 
We consider grid points $\Gamma$ given by
\begin{align*}
  \Gamma \defeq \left\{\frac{1}{M}(\ell_1 r_1, \ell_2 r_2,\dots, \ell_k r_k)~~|~~ \ell_1,\dots,\ell_k\in \{0,1,\dots,M\} \right\}.
\end{align*}
For an arbitrary $\Vec{u} \in \Gamma$, let
\begin{align*}
S_j(\Vec{u}) \defeq \{s\in [-1,1]| \exists i\in \{1,\dots,k\},\exists \ell \in \Integer ~~\text{s.t.}~~ u_j - |s-p_{i,j}|=\ell r_i / M\},
\end{align*}
for $j=1,\dots,k$.
Then let 
\begin{align*}
  T_j(\Vec{u}) \defeq \bigcup_{j=1}^k S_j(\Vec{u}) \cup \{-1,1\}.
\end{align*}
Suppose $t_0,t_1,\dots,t_m$ be an ordering of all elements of $T_i(u,v)$ such
that $t_i\le t_{i+1}$ for any $i=0,1,\dots,m$.
Then we can compute $G_i(\Pi, \Vec{u})$ for any $\Vec{u}\in \Gamma$ by 
\begin{align}
  G_j(\Pi,\Vec{u}) &=\overline G_i(\Pi, \Vec{u}) \nonumber \\
  &= \int_{-1}^1 G_{j-1}(\Pi, \Vec{u}- \Vec{q}_j(s)){\rm d}s \nonumber \\
  &= \sum_{i=0,\dots,m} \int_{t_i}^{t_{i+1}} G_{j-1}\left(\Pi, \frac{1}{M} \Vec{w}(\Vec{u}, t_{i+1})\right) {\rm d}s  && \mbox{(by \eqref{form:G_k})}\nonumber \\
  &= \sum_{i=0,\dots,m} (t_{i+1}-t_i)G_{j-1}\left(\Pi, \frac{1}{M} \Vec{w}(\Vec{u}, t_{i+1})\right), \label{form:Gbar2_k}
\end{align}
where $\Vec{w}(\Vec{u},t_{i+1})=(\lceil M(u_1-|t_{i+1}-p_{1,j}|)\rceil ,\dots, \lceil M(u_k-|t_{i+1}-p_{k,j}|)\rceil)$.


Our algorithm outputs the value of $G_n(\Pi, \Vec{r})$.
By taking the parameter $M$ larger, 
we get closer approximation of ${\rm Vol}(S(\Pi,\Vec{r}))$.
Here we assume that $0< \delta\le 1/2$.
The following is our algorithm \ref{algorithm:CPI_k}.
\begin{algorithm}
\label{algorithm:CPI_k}
  Input: $\Pi\in \Real^{kn}$, $\Vec{0}\le \Vec{r} \le \Vec{1}$; \\
  1. Let $M:=2kn^2/\delta$; \\
  2. Let $G_0(\Pi,\Vec{u}):=1$ for $\Vec{u}\ge \Vec{0}$, otherwise $G_0(\Pi,\Vec{u}):=0$; \\
  3. For $j:=1,\dots,n$,\\
  4. \hspace*{3mm} Compute $\overline{G}_j(\Pi,\Vec{u})$ from $G_{j-1}(\Pi,\Vec{u})$ by (\ref{form:Gbar2_k}); \\
  5. \hspace*{3mm} Compute staircase approximation $G_j(\Pi,\Vec{u})$ of $\overline{G}_j(\Pi,\Vec{u})$ by (\ref{form:G_k}); \\
  6. Output $G_n(\Pi,\Vec{r})$.
\end{algorithm}

Let us consider the running time of our algorithm \ref{algorithm:CPI_k}.
In Step 4-5, computing $\overline {G}_j(\Pi, \Vec{u})$ for a fixed $\Vec{u}$ takes $O(kM)$ time
because $\overline {G}_j(\Pi, \Vec{u})$ is the sum of $m\le 2kM$ values.
We compute $\overline {G}_j(\Pi, \Vec{u})$  for $(M+1)^k$ different $\Vec{u}$'s.
Then Step 4-5 is repeated $n$ times.
We have the following observation.
\begin{observation}
  The running time of Algorithm \ref{algorithm:CPI_k} is $O(knM^{k+1})$.
\end{observation}

\subsection{Proof of Theorem \ref{th:CPI_k}}
Here, we prove that $M=2kn^2/\delta$ is sufficient to 
have $1+\delta$ approximation of ${\rm Vol}(S(\Pi,\Vec{r}))$.
We show the following lemma.
\begin{lemma}
$\Psi_j(\Pi,\Vec{u})$ is non-decreasing with respect to each component of $\Vec{u}$.
\end{lemma}
\begin{proof}
Let $\Vec{u}=(u_1,\dots,u_k)\le \Vec{u}'=(u'_1,\dots,u'_k)$.
By definition, we have that
\begin{align*}
  \Psi_j(\Pi,\Vec{u}) &={\rm Vol}\left(\left\{\Vec{x}\in \Real^j \left| \bigwedge_{i=1}^k \sum_{\ell=1}^j |x_\ell-p_{i,\ell}| \le u_i \right.\right\}\right)  \\
  &\le {\rm Vol}\left(\left\{\Vec{x}\in \Real^j \left| \bigwedge_{i=1}^k \sum_{\ell=1}^j |x_\ell- p_{i,\ell}| \le u'_i \right.\right\}\right) = \Psi_j(\Pi,\Vec{u}').
\end{align*}
\end{proof}

Then, we can prove the following lemma, which gives upper and lower bounds 
on the approximation $G_n(\Pi,\Vec{u})$.
\begin{lemma}
\label{lemma:sandwitch_k}
$\Psi_n(\Pi,\Vec{u}) \le G_n(\Pi,\Vec{u}) \le \Psi_n(\Pi,\Vec{u}+n\Vec{r}/M)$.
\end{lemma}
\begin{proof}
Since $\Psi_n(\Pi,\Vec{u})\le G_n(\Pi,\Vec{u})$ is clear from the algorithm,
we prove $G_n(\Pi,\Vec{u})\le \Psi_n(\Pi,\Vec{u}+n\Vec{r}/M)$ in the following.
This is proved by induction on $n$.
Since $G_0(\Pi,\Vec{u})=\Psi_0(\Pi,\Vec{u})$ for any $\Vec{u}\in \Real_{\ge 0}^k$, 
the base case holds.
Then, as for the induction step, we assume 
$G_j(\Pi,\Vec{u}) \le \Psi_j(\Pi,\Vec{u}+j\Vec{r}/M)$.
By the definition of $G_j(\Pi,\Vec{u})$ and $\overline{G}_{j+1}(\Pi,\Vec{u})$, 
we have
\begin{align*}
  G_{j+1}(\Pi, \Vec{u}) &= \overline G_{j+1} (\Pi, \lceil M\Vec{u}/\Vec{r}\rceil/M) \\
  &= \int_{-1}^1 G_j(\Pi, \lceil M \Vec{u}/\Vec{r}\rceil/M -\Vec{q}_{j+1}(s)){\rm d}s \\
  &\le \int_{-1}^1 \Psi_j(\Pi, \lceil M \Vec{u}/\Vec{r}\rceil/M-\Vec{q}_{j+1}(s)+j\Vec{r}/M){\rm d}s && \mbox{(Induction hypo.)}\\
  &\le \int_{-1}^1 \Psi_j(\Pi, \Vec{u}/\Vec{r} -\Vec{q}_{j+1}(s)+(j+1)\Vec{r}/M){\rm d}s \\
  &= \Psi_{j+1}(\Pi, \Vec{u}/\Vec{r}+(j+1)\Vec{r}/M),
\end{align*}
where $\Vec{u}/\Vec{r}=(u_1/r_1, \dots,u_k/r_k)$ and $\lceil\Vec{u}/\Vec{r}\rceil=(\lceil u_1/r_1 \rceil, \dots, \lceil u_k/r_k \rceil)$.
Then we have the lemma.
\end{proof}

We prove Theorem \ref{th:CPI_k} as follows.

\noindent{\it Proof of Theorem \ref{th:CPI_k}}:
By Lemma \ref{lemma:sandwitch_k}, we have that the approximation ratio
is bounded from above by 
$\Psi_n(\Pi,\Vec{u}+\Vec{h})/\Psi_n(\Pi,\Vec{u})$, where $\Vec{h}=(h_1,\dots,h_k)\le n\Vec{r}/M$.
We bound the reciprocal of the approximation ratio from below.

For convenience, let
\begin{align*}
  K_i(\Pi,\Vec{u},d)=\left\{\Vec{x}\in \Real^n \big| \exists\Vec{y}\in S(\Pi,\Vec{u}), \exists b\in [0,d], \Vec{x}-\Vec{p}_i=b(\Vec{y}-\Vec{p}_i),\text{s.t.} \|\Vec{y}-\Vec{p}_i\|_1=u_i \right\}.
\end{align*}
Here $K_i(\Pi,\Vec{u},1)$ is given by considering the cones that are 
given by the center $\Vec{p}_i$ as the top vertex and 
the shared surface of $S(\Pi,\Vec{u})$ and $C(\Vec{p}_i,u_i)$ as the bottom.
Then $K_i(\Pi,\Vec{u},d)$ is given by scaling $K_i(\Pi,\Vec{u},1)$.
Since we assume that $\Vec{0}\in S(\Pi,\Vec{u})$, 
we have $K_i(\Pi,\Vec{u},1)\subseteq S(\Pi,\Vec{u})$.
Since ${\rm Vol}(S(\Pi,\Vec{u})-K_i(\Pi,\Vec{u},1))$ is 
equal to ${\rm Vol}(S(\Pi,\Vec{u}+h_i\Vec{e}_i)-K_i(\Pi,\Vec{u},(u_i+h_i)/u_i))$,
we have that
\begin{align*}
&\frac{\Psi_n(\Pi,\Vec{u})}{\Psi_n(\Pi,\Vec{u}+h_i\Vec{e}_i)} =\frac{{\rm Vol}(S(\Pi,\Vec{u}))}{{\rm Vol}(S(\Pi,\Vec{u}+h_i\Vec{e}_i))} \\
&=\frac{{\rm Vol}(S(\Pi,\Vec{u})-K_i(\Pi,\Vec{u},1))+{\rm Vol}(K_i(\Pi,\Vec{u},1))}{{\rm Vol}(S(\Pi,\Vec{u}+h_i\Vec{e}_i)-K_i(\Pi,\Vec{u},(u_i+h_i)/u_i))+{\rm Vol}(K_i(\Pi,\Vec{u},(u_i+h_i)/u_i))}\\
&\ge \frac{{\rm Vol}(K(\Pi,\Vec{u},1))}{{\rm Vol}(K_i(\Pi,\Vec{u},(u_i+h_i)/u_i))}
\ge \frac{1}{(1+h_i/u_i)^n}.
\end{align*}

This leads to 
\begin{align*}
  \frac{\Psi_n(\Pi,\Vec{r})}{\Psi_n(\Pi,\Vec{r}+\Vec{h})}\ge \prod_{i=1}^k\frac{1}{(1+h_i/r_i)^n}\ge 1-\sum_{i=1}^k nh_i/r_i.
\end{align*}

Then, for $\delta \le 1/2$, we have $\frac{\Psi_n(\Pi,\Vec{r}+\Vec{h})}{\Psi_n(\Pi,\Vec{r})} \le \frac{1}{1-\sum_{i=1}^k nh_i/r_i}\le \frac{1}{1-kn^2/M} \le 1+2kn^2/M = 1+\delta$.
\qed

\section{The Volume of ${\cal V}$-polytopes with $n+k$ Vertices}\label{sec:n+c}
Given a vertex set $V=\{\Vec{v}_1,\dots,\Vec{v}_{n+k}\}$, where 
$k\ge 1$ is a constant.
Here we consider the problem of computing the volume of 
$P={\rm conv}(V)$.
Without loss of generality, we assume that $P$ contains the origin $\Vec{0}$
as its interior point. 
Also note that we assume that all the vectors are vertical vectors.
Then we have the following Theorem.
\begin{theorem}
  By decomposing $P$ into simplices, we can compute ${\rm Vol}(P)$ in $O(n^{k+3})$ time.
\end{theorem}

The following is the algorithm for computing ${\rm Vol}(P)$.
For all possible $U\subseteq V$, we check if the $n-1$ dimensional 
polytope $f_U$ given by $U$ is the facet of $P$, 
and if so, we compute the volume $S_U:=\det(M_U)/n!$,
where $M_U=(\Vec{u}_1, \dots, \Vec{u}_n)$.
Then ${\rm Vol}(P)=\sum_{U\subseteq V}S_U$.
\begin{algorithm}
Input: $V=\{\Vec{v}_1,\dots,\Vec{v}_{n+k}\}\in \Real^{n(n+k)}$\\
1. $S:=0$, $M_V:=(\Vec{v}_1,\dots,\Vec{v}_{n+k})$;\\
2. For all possible $U=\{\Vec{u}_1,\dots,\Vec{u}_n\}\subseteq V$,\\
3. \hspace*{3mm} Compute $\Vec{a}=(a_1,\dots,a_n)\in \Real^n$ s.t. $\Vec{a}^{\top}M_U=\Vec{1}$,where $M_U=(\Vec{u}_1,\dots,\Vec{u}_n)$; \\
4. \hspace*{3mm} If $\Vec{a}^{\top}M_V\le \Vec{1}$ or $\Vec{a}^{\top}M_V\ge \Vec{1}$, then \\
5. \hspace*{6mm}     $S:=S+{\rm Vol}(S_U)$, where ${\rm Vol}(S_U)=\det(M_U)/n!$; \\
6. Output $S$.
\end{algorithm}

We consider the running time of the algorithm.
The loop from Step 2 to Step 5 is repeated 
$\binom{n+k}{k}$ times.
In Step 3, we compute $\Vec{a}$ by the Gaussian elimination,
which takes $O(n^3)$ time.
Step 4 checks if all vertices is contained in a half space given by $f_U$.
This takes at most $n(n+k)$ additions and multiplications.
In Step 5, computing ${\rm Vol}(S_U)$ takes $O(n^3)$.
The running time amounts to 
$O\left((n^3+n(n+k)+n^3)\binom{n+k}{k}\right)=O(n^{k+3})$.

\section{Conclusion}
 Motivated by a deterministic approximation of the volume of a ${\cal V}$-polytope, 
  this paper gave an FPTAS for the volume 
of the knapsack dual polytope $\Vol(P_{\Vec{a}})$. 
 In the process, 
  we showed that the volume of the intersection of $L_1$-balls is {\#}P-hard, and 
  gave an FPTAS.
 As we remarked, 
  the volume of 
the intersection of two $L_q$-balls are easy for $q=2,\infty$. 
 The complexity 
  of the volume of 
the intersection of two $L_q$-balls for other $q>0$ is interesting. 
 The problem seems difficult even for approximation in the case of $q \in (0,1)$, 
  since $L_q$-ball is no longer convex. 
 Our FPTAS for the intersection of two cross-polytopes 
  assumes that each cross-polytope contains the center of the other one. 
 It is open if an FPATS exists without the assumption. 

 We have remarked that the volume of a ${\cal V}$-polytope with $n+k$ vertex is computed in $\Order(n^{k+3})$, 
  while Khachiyan's result~\cite{Khachiyan1989} implies that it is {\#}P-hard when $k \geq n+1$. 
 The complexity when $k=\omega(1)$ and $k=\order(n)$ seems not known. 
 It is an interesting question if an FPT algorithm regarding $k$ exists.

\section*{Acknowledgments}
 This work is partly supported by
  Grant-in-Aid for Scientific Research on Innovative Areas MEXT Japan
 ``Exploring the Limits of Computation (ELC)'' (No. 24106008, 24106005).


\end{document}